\documentclass[a4paper,11pt]{article}
\usepackage{amsfonts}
\usepackage{booktabs}
\usepackage{caption}
\usepackage{hyperref}
\usepackage[nodayofweek]{datetime}
\usepackage{enumitem}
\usepackage{float}
\usepackage[margin=2.5cm]{geometry}
\usepackage{graphicx}
\usepackage[utf8]{inputenc}
\usepackage{numbertabbing}
\usepackage{times}
\usepackage{url}
\usepackage{xcolor}
\usepackage{xspace}
\usepackage{mdframed} 
\usepackage{tikz}
\usepackage{xcolor}
\usepackage{ifthen}
\usepackage{framed}
\usepackage{amsmath} 
\usepackage{cleveref}

\usepackage[commandnameprefix=ifneeded,commentmarkup=uwave]{changes}
\definecolor{duc}{rgb}{0.0, 0.38, 0.67}	
\definechangesauthor[name=duc, color=duc]{duc}
\definechangesauthor[name=sarah, color=violet]{sarah}
\definecolor{gcolor}{rgb}{0.9, 0.17, 0.31}
\definechangesauthor[name=luca, color=gcolor]{luca}
\definechangesauthor[name=christian, color=red]{christian}
\definechangesauthor[name=marko, color=red]{marko}
\definechangesauthor[name=guy, color=cyan]{guy}

\newcommand{\aprob}{\ensuremath{{\varrho}_\CA}\xspace}
\newcommand{\hprob}{\ensuremath{{\varrho}_H}\xspace}
\newcommand{\prob}{\ensuremath{{\varrho}}\xspace}
\newcommand{\dprob}{\ensuremath{{\varrho}}\xspace}

\newcommand{\powork}{\text{PoW}\xspace}

\newcommand{\postake}{\text{PoS}\xspace}
\newcommand{\PoStake}{\text{Proof-of-Stake}\xspace}
\newcommand{\postorage}{\text{Proof-of-Storage}\xspace}
\newcommand{\PoStorage}{\text{Proof-of-Storage}\xspace}

\newcommand{\allocator}{\op{RA}\xspace}
\newcommand{\gR}{\ensuremath{\allocator_\mathsf{res}}\xspace}
\newcommand{\powR}{\ensuremath{\allocator_\mathsf{pow}}\xspace}
\newcommand{\stakeR}{\ensuremath{\allocator_\mathsf{pos}}\xspace}
\newcommand{\storageR}{\ensuremath{\allocator_\mathsf{space}}\xspace}
\newcommand{\cost}{\op{Cost}\xspace}

\newdateformat{simple}{\THEDAY\ \monthname[\THEMONTH]\ \THEYEAR}
\simple

\floatstyle{ruled}
\newfloat{algo}{htbp}{algo}
\floatname{algo}{Algorithm}

\def \ifempty#1{\def\temp{#1} \ifx\temp\empty }

\newcommand{\str}[1]{\textsc{#1}}
\newcommand{\var}[1]{\textit{#1}}
\newcommand{\op}[1]{\textsl{#1}}
\newcommand{\msg}[2]{\ensuremath{\ifempty{#2} [\str{#1}] \else [\str{#1}, {#2}] \fi}}

\newcommand{\becomes}{\ensuremath{\leftarrow}}

\newcommand{\false}{\textsc{false}\xspace}
\newcommand{\true}{\textsc{true}\xspace}
\newcommand{\etal}{\emph{et al.}}

\newcommand{\BC}{\ensuremath{\mathbb{C}}\xspace}

\newcommand{\BN}{\ensuremath{\mathbb{N}}\xspace}

\newcommand{\BP}{\ensuremath{\mathbb{P}}\xspace}

\newcommand{\CA}{\ensuremath{\mathcal{A}}\xspace}
\newcommand{\CB}{\ensuremath{\mathcal{B}}\xspace}
\newcommand{\CC}{\ensuremath{\mathcal{C}}\xspace}
\newcommand{\CD}{\ensuremath{\mathcal{D}}\xspace}

\newcommand{\CP}{\ensuremath{\mathcal{P}}\xspace}

\usepackage{todonotes}
\setuptodonotes{fancyline, inline, color=red!30}

\newcommand{\localBlocks}{\ensuremath{\mathcal{B}}\xspace}
\newcommand{\chain}{\ensuremath{\mathcal{C}}\xspace}
\newcommand{\prechain}{\ensuremath{\mathcal{C}_{\mathsf{prefix}}}\xspace}

\newcommand{\bstate}{\ensuremath{\var{st}}\xspace}
\newcommand{\slot}{\ensuremath{\var{sl}}\xspace}
\newcommand{\stakeDis}{\ensuremath{\mathbb{S}_{\mathsf{stake}}}\xspace}
\newcommand{\spaceDis}{\ensuremath{\mathbb{S}_{\mathsf{space}}}\xspace}

\newcommand{\lchain}{\ensuremath{\mathcal{C}_{\mathsf{local}}}\xspace}
\newcommand{\btx}{\ensuremath{B_{\mathsf{com}}}\xspace}
\newcommand{\pparagraph}[1]{\smallskip \noindent \textbf{#1.}}

\newcommand{\fun}[1]{\ensuremath{\op{#1}}}

\newcommand{\sample}{\ensuremath{\xleftarrow{R}}\xspace}

\newcommand{\PP}{\ensuremath{\mathbb{P}\xspace}}

\newcommand{\tx}{\mathsf{tx}\xspace}

\newcommand{\set}[1]{\ensuremath{\{#1\}}\xspace}

\newcommand{\sign}{\fun{Sign}\xspace}
\newcommand{\verify}{\fun{Verify}\xspace}

\newcommand{\bset}{\ensuremath{\{0,1\}}}

\usepackage{amsthm}
\newtheoremstyle{plain-boldhead}
  {\topsep}
  {\topsep}
  {\itshape}
  {}
  {\bfseries}
  {.}
  { }
  {\thmname{#1}\thmnumber{ #2}\thmnote{ (\bfseries #3)}}
\newtheoremstyle{definition-boldhead}
  {\topsep}
  {\topsep}
  {\normalfont}
  {}
  {\bfseries}
  {.}
  { }
  {\thmname{#1}\thmnumber{ #2}\thmnote{ (\bfseries #3)}}
\theoremstyle{plain-boldhead}
\newtheorem{theorem}{Theorem}
\newtheorem{proposition}[theorem]{Proposition}
\newtheorem{lemma}[theorem]{Lemma}

\theoremstyle{definition-boldhead}
\newtheorem{definition}{Definition}
\newtheorem{remark}{Remark}

\providecommand{\keywords}[1]
{
  \small    
  \textbf{Keywords.} #1
}

\begin{document}
\title{\bf Modeling Resources in Permissionless Longest-chain Total-order Broadcast}
\author{Sarah Azouvi\\
  Protocol Labs
  \and Christian Cachin\\
  University of Bern
%
  \and Duc V. Le\\
  University of Bern
  %
  \and Marko Vukoli\'c\\
  Protocol Labs
  \and Luca Zanolini\\
   University of Bern
}

\maketitle

\begin{abstract}\noindent
Blockchain protocols implement total-order broadcast in a
permissionless setting, where processes can freely join and leave.
In such a setting, to safeguard against Sybil attacks, correct processes rely
on cryptographic proofs tied to a particular type of \emph{resource}
to make them eligible to order transactions. For example, in the case of
Proof-of-Work (PoW), this
resource is computation, and the proof is a solution to a computationally hard
puzzle. Conversely, in Proof-of-Stake (PoS), the resource corresponds
to the number of coins that every process in the system owns, and a
secure lottery selects a process for participation proportionally to its
coin holdings.

Although many resource-based blockchain protocols are formally proven
secure in the literature, the existing security proofs fail to demonstrate
why particular types of resources cause the blockchain protocols to be
vulnerable to distinct classes of attacks. For instance, PoS systems are
more vulnerable to long-range attacks, where an adversary corrupts past
processes to re-write the history, than \powork and \PoStorage systems.
Proof-of-Storage-based and \postake-based protocols are both more
susceptible to private double-spending attacks than \powork-based
protocols; in this case, an adversary mines its chain in secret without
sharing its blocks with the rest of the processes until the end of the
attack.

In this paper, we formally characterize the properties of resources through an
abstraction called \emph{resource allocator} and
give a framework for understanding longest-chain consensus protocols based on
different underlying resources. In addition, we use this resource allocator to 
demonstrate security trade-offs
between various resources focusing on well-known attacks (e.g., the long-range
attack and nothing-at-stake attacks).

\end{abstract}

\newpage

\section{Introduction}
\label{sec:intro}

Permissionless consensus protocols are open for everyone to participate and often rely on a \emph{resource} to protect against Sybil attacks. 
In the case of Proof-of-Work (PoW), this resource is computation: A computational puzzle must be solved in order to gain writing rights in the system.
In contrast, in a Proof-of-Stake (PoS) system, writing access is granted using a form of lottery where participants are elected proportionally to the number of coins they own.
Other resource-based systems, such as \postorage, have also appeared. Participants are elected proportionally to the number of resources they commit to the system, and hence this commitment must be publicly verifiable. 
Different resources present different trade-offs. For example, \postake is much more energy-efficient than \powork but presents many additional vulnerabilities~\cite{DBLP:conf/ec/Brown-CohenNPW19}.
Comparing the security of protocols based on multiple resource types is a non-trivial task, as they all operate under different assumptions and frameworks. 

In this paper, we provide a common framework to formally compare consensus protocols based on different underlying resources. 
We only consider longest-chain protocols~\cite{DBLP:conf/ccs/DemboKTTVWZ20} based on an underlying resource, as we want to highlight the properties affected by relying varying the resource for the same consensus method.
In future work, our framework could be used to model further approaches to ensure consensus, such as well-known BFT protocols~\cite{DBLP:books/daglib/0025983}, for instance. In longest-chain protocols, one participant is elected at each time step, on expectation, proportionally to their amount of resource and that participant gets to write to the append-only database by adding a \emph{block} containing all the necessary data to the longest chain of blocks.

We also explore known attacks in this work. The first one is the long-range attack. 
In a long-range attack, an adversary corrupts processes that used to participate in the system but that no longer hold any resources.
Moreover, we investigate nothing-at-stake attacks, where processes mine on multiple chains at the same time, and private attacks, where an adversary mines on its own chain without contributing to the honest chain.
We are interested in quantifying gain or loss of security with different resources.
It has already been shown that, when considering longest-chain protocols, \postake is less secure than \powork. We furthermore show that \postorage stands in the middle, as storage is not virtual (like stake), but is reusable (unlike the computation of \powork).

We start the paper by providing a formal framework in which protocols based on different resources can meaningfully be compared.
We differentiate between virtual and external resources to highlight which properties make longest-chain \postake and \postorage less secure than \powork, although they both present trade-offs when it comes to their efficiency.

\pparagraph{Contributions} Our contributions can be summarized as follows.
\begin{itemize}
    \item We formally characterize the properties of resources through an
abstraction called \emph{resource allocator} and formally define properties for a \emph{secure} resource allocator.
    \item We concretely define different resource allocator abstractions, each one for every type of resource used in popular blockchain protocols, namely, computation, stake, and storage. 
    \item We present an algorithm that, when instantiated with different resource allocators, leads to a generalization of existing protocols such as Nakamoto consensus, Ouroboros Praos, and Filecoin's consensus protocol. We also formally show this generalization implements total-ordered broadcast under a fixed total resource and permissionless setting. 
    \item We demonstrate how different resources lead to different security trade-offs by leveraging our model to explain long-range attacks against virtual resources and attacks related to the nothing-at-stake nature of reusable resources.
\end{itemize}

\pparagraph{Related Work}
Since the emergence of Bitcoin in 2008, the academic community has
developed a number of frameworks~\cite{DBLP:conf/eurocrypt/GarayKL15,DBLP:journals/iacr/Ren19,DBLP:conf/ccs/GaziKR20,DBLP:conf/ccs/DemboKTTVWZ20} 
for studying the safety and liveness properties of its Nakamoto consensus protocol.
These studies also established a strong foundation for the development of  blockchain 
protocols based on more eco-friendly types of resources, such as stake and storage. 
However, despite the fact that all resource-based blockchains have been formally proven to be secure, they failed to explain why certain properties of resources make some blockchain protocols more susceptible to particular types of attacks than others. 
To the best of our knowledge, no prior work has attempted to formally study the properties of underlying resources, and our work aims to fill this gap.

Lewis-Pye and Roughgarden~\cite{DBLP:journals/corr/abs-2101-07095} present the concept of
a \emph{resource pool} that reflects the resource balance of processes in the system
at any time, and they use a \emph{permitter} together with the resource pool to abstract away the leader selection procedure. 
Using this formalization, they demonstrate two crucial impossibility results for permissionless systems. 
Two main results of their work are:
\emph{(i)} no permissionless, deterministic, and decentralized protocol solves the Byzantine Agreement problem in a synchronous setting, 
and \emph{(ii)} no permissionless and {probabilistic} protocols solve the Byzantine Agreement problem in the unsized setting (in which the total number of resources is unknown) with partially synchronous communication.
However, their work could not capture several aspects of underlying resources used in blockchain protocols;
therefore, their work did not demonstrate long-range attacks against virtual resources such as stake, and
the cost of several other attacks on \emph{reusable} resources.
Our work takes a similar approach of abstracting away the leader selection process with 
a resource allocator (c.f., \cref{sec:modeling_resources}), and 
we further formalize the properties of resources through the interactions between the process and this allocator. 
With this formalization, we prove how permissionless and {probabilistic}
blockchain protocols guarantee properties of a total-order broadcast in a
synchronous setting and demonstrate
various attacks against \emph{virtual} or \emph{reusable} resources.

Terner~\cite{DBLP:conf/fc/Terner22} also investigates how to abstract
resources used in permissionless blockchains. While their work outlines
several essential properties of resources and studies how the resource
generation rate affects the standard properties (i.e., consistency and
liveness) of robust transaction ledger, their study does not characterize the
properties of the underlying resources used in permissionless blockchain
protocols. Consequently, their model fails to explain why distinct types of
resources render blockchain protocols vulnerable to different attacks (e.g.,
long-range attacks and private attacks).

\section{Model and Definitions}%
\label{sec:model_and_definitions}

\subsection{System Model}%
\label{sub:system_model}
\pparagraph{Time} We assume that the protocol proceeds in \emph{time steps} and define a time step to be a value in $\BN$. Moreover, we consider $0$ as starting time step of protocol execution.

\pparagraph{Processes} 
{We consider a system consisting of a set of \emph{processes}, $\CP =\{p_1, p_2,
\ldots \}$}. 
Processes interact with each other
through exchanging messages. A protocol for $\CP$ consists of a collection of
programs with instructions for all processes. 
Moreover, to capture the permissionless nature of various blockchain protocols, processes can join
the system at any time.
we denote, $\CP_{\leq t}$, the set of all processes that have participated in the
protocol before the time step $t$. Hence, $\CP_{t} \subseteq \CP_{t'}$
for all $t\leq t'$.
{At the beginning of each time step, a process becomes \emph{activated}, and it starts to follow a deterministic protocol. This includes processing any messages that may have arrived from other processes. Once done, it becomes \emph{deactivated}. We assume that the activation period of a process $p_i$ starts at the time step $t$ and ends before time step $t+1$.}

\pparagraph{Communication} 
{We assume there is a low-level primitive for sending messages over point-to-point links between each pair of processes {that know of each other}, as well as a probabilistic broadcast primitive~\cite{DBLP:books/daglib/0025983}. Point-to-point messages are authenticated and delivered reliably among correct processes.
In probabilistic broadcast, correct processes \op{gossip-deliver} and \op{gossip-broadcast} messages with an overwhelming probability, no message is delivered more than once, and no message is created or corrupted by the network.}

\pparagraph{Network Delay} We denote by {$\Delta \in \mathbb{N}$} with $\Delta \ge 1$ the maximum network delay~\cite{DBLP:journals/jacm/DworkLS88}. 
Namely, if a correct process \op{gossip-broadcasts} 
a message $m$ at a time step
$t$, then other processes will have \op{gossip-delivered} or received over the message by the beginning
of a time step $t+\Delta$ with an overwhelming probability.

\pparagraph{Idealized Digital Signature} 
A digital signature scheme, $\Sigma$, consists of two
operations, $\sign(\cdot, \cdot)$ and $\verify(\cdot, \cdot, \cdot)$.
\sloppy{The operation $\sign(p_i, \cdot)$ invoked by $p_i$ takes  $m \in
\bset^*$ as input and returns a signature $\sigma \in \bset^*$.}
Only $p_i$ can invoke $\sign(p_i, \cdot)$.
The operation $\verify(p_i, \cdot, \cdot)$ takes as input a signature, $\sigma$,
and a message $m$;  $\verify(p_i, \cdot, \cdot)$ returns \true for any
$p_i \in \mathcal{P}$ and $m\in \bset^*$ if and only if $p_i$ has invoked 
 $\sign(p_i, m)$ and obtained  $\sigma$ before. Any process can invoke $\verify (\cdot, \cdot, \cdot)$.

\pparagraph{Random Oracle} 
All hash functions are modeled as a random oracle, $H$, that can be queried by any process.
$H$ takes as input a bit string $x\in \bset^*$ and returns a uniformly random
string from $\bset^\lambda$ where $\lambda$
is the security parameter. 
Also, upon repeated queries, $H$ always outputs the same answer.

\subsection{Modeling Blockchain Data Structures}
\label{sub:model-blockchain}

\pparagraph{Blocks} We use $\tx$ to denote a \emph{transaction}. 
We write $\overline{\tx} = [\tx_1, \dots, \tx_m]$ to denote a list of transactions. A \emph{block} is $B = (h, \overline{\tx}, \pi, \sigma_i)$, 
where $h$ is a hash value, 
$\overline{\tx}$ is a list of transactions, 
$\pi$ is a resource commitment proof (cf. Section~\ref{sec:modeling_resources})
and $\sigma_i$ is a signature on $(h,\overline{\tx}, \pi)$. 
In this work, we assume that blocks are signed. In this way, we can abstract away the notion of \emph{coinbase} transactions, i.e., the first transaction in a block, created by a miner, and used to collect the block reward.
{Finally, we denote with $B_0 = (\bot,
\overline{\tx}, \bot, \bot)$ the {\em genesis block}.}

\pparagraph{Blockchain} A \emph{blockchain} $\chain = [B_0, B_1, \dots]$ with
respect to the genesis block $B_{0}$ is a chain of blocks forming a hash chain
such that $h_j = H(B_{j-1})$ {for $h_j \in B_j$ for $j = 1, 2, \dots$ with $B_j = (h_j, \bar{\tx_j}, \pi_j, \sigma_j)$}. 
For a blockchain $\chain$, we use $\chain[-k]$ to denote the last $k$-th block in $\chain$, 
let $\chain[k]$ to denote block $B_k$ (i.e., block at height $k$), 
and write $\chain[:-k]$ to denote
the first $|\chain|-k$ blocks. 
$|\chain|$ denotes the length of $\chain$.
We write $\chain \preceq \chain'$ when $\chain$ is a prefix $\chain'$.
We use $\chain^{t}$ to denote the blockchain at time step $t$. 
For two time steps, $t_1$ and $t_2$, $\chain^{t_2}/\chain^{t_1}$ is a set of blocks that is in $\chain^{t_2}$
but not in $\chain^{t_1}$.

\pparagraph{State} 
{The blockchain \emph{state} \bstate specifies different information
of the underlying blockchain protocol, e.g.,  
the stake distribution of each process, the block information, such as timestamps, 
as well as contract local states.  
The blockchain state \bstate can be reconstructed by executing transactions  
included in a blockchain $\chain$. Without loss of generality, we define
the state to be the blockchain, $\bstate = \chain$. 
Also, we write $\bstate = (\chain, B)$ to indicate that a block $B$ is potentially appended to $\chain$.}

\pparagraph{Validity} We introduce the notion of {\em validity} for transactions and
blockchains to capture the fact that only ``valid'' transactions are {delivered}. 
More importantly, for all blockchain protocols, the decision on the validity is determined locally by all processes. Because of this, we define the validity as follows.
A transaction $x$ is \emph{valid} with respect to $\chain$
if $\tx$ satisfies a \emph{validation predicate} $\mathbb{P}(\chain, \cdot)$
locally known to all processes (i.e., $\PP(\chain, [\tx]) = \true$). 
We also use $\PP(\chain, \overline{\tx})=\true$ to indicate 
{that the sequence of transactions in $\overline{\tx}$ is valid
(i.e., does not consume the same output in
Bitcoin or the same nonce in Ethereum), and we define $\BP(\chain, [~])$ to be
$\true$}.
Depending on the blockchain protocol, a \emph{valid block} $B$ issued by
$p_i$ should consist of 
a valid signature issued by $p_i$, a valid ``proof'' $\pi$
for a so-called resource commitment that we introduce in  Section ~\ref{sec:modeling_resources}
and valid transactions with respect to $\chain$ such that $\chain[-1]=B$,
(i.e., $\BP(\chain, \overline{\tx})=\true$ for $\overline{\tx} \in B$).
Finally, \emph{valid blockchains} are chains that consist of 
only valid blocks
and start from the genesis block $B_0$.

\subsection{Total-order Broadcast}
{We will show that the blockchain protocols considered here guarantee the following properties of total-order broadcast in a permissionless setting.}
In particular, total-order broadcast ensures that all processes deliver the same set of
transactions in a common global order.
In total-order broadcast,  every process broadcasts a transaction by invoking
$\op{a-broadcasts}(\tx)$. 
The broadcast primitive outputs a transaction $\tx$ through an $\op{a-deliver}(\tx)$
event.  
{In this model, we do not distinguish between a process and a client. 
A client can be considered as a process that only broadcast transactions and does not
participate in mining.}

\begin{definition}[Total-order Broadcast]
\label{def:tob}

A protocol for total-order broadcast satisfies the following properties.
  \begin{description}[style=unboxed]
      \item[Validity] If a correct process, $p_i$ \op{a-broadcasts} a valid
        transaction $\tx$ according to $\mathbb{P}(\cdot, \cdot)$ (i.e., {the
        validation predicate defined in~\Cref{sec:model_and_definitions}}),
        then $p_i$ eventually \op{a-delivers} $\tx$ with an overwhelming probability.
      \item[No duplication] No correct process \op{a-delivers} the same
        transaction $\tx$ more than once.
      \item[Agreement] If a transaction $\tx$ is \op{a-delivered} by some correct
        process, then with an overwhelming probability $\tx$ is eventually
        \op{a-delivered} by every correct process.
      \item[Total order] Let $\tx_1$ and $\tx_2$ be any two transactions, and suppose
        $p_i$ and $p_j$ are any two correct processes that \op{a-deliver}
        $\tx_1$ and $\tx_2$. If $p_i$ \op{a-delivers} $\tx_1$ before $\tx_2$, then
        with an overwhelming probability, 
        $p_j$ \op{a-delivers} $\tx_1$ before $\tx_2$.
  \end{description}
\end{definition}

\section{Modeling Resources in Blockchain}
\label{sec:modeling_resources}
{In this section, we model resources, formalize their properties through the abstraction of a \emph{resource allocator}, and state our threat assumptions. {The definition of a \emph{resource allocator} in this section is only syntactic; security and liveness properties of the resource allocator are defined in~\Cref{sec:nakamoto_consensus}}}. 
\begin{definition}[Resource Budget]
  A resource budget $r$ is a value in \BN. 
  At any given time, each process $p_i$ has a resource budget $r_i$. 
  In particular, 
  {
  there exists a function $\fun{Alloc}:\mathcal{P}\times \BN
  \rightarrow \mathbb{N}$ that takes as input a process $p_i$ and
  a time step $t$, outputs the resource budget of a process at time step $t$.
  }
  We define $R$ to be the fixed  resource budget
  existing in the system. 
\end{definition}
The definition of a fixed resource budget and the resource allocation function can be viewed as the sized setting and the resource pool definition in Lewis-Pye and Roughgarden framework~\cite{DBLP:journals/corr/abs-2101-07095}.
We note that the specification of the resource budget varies depending on protocols;
e.g., for \powork, we define the budget to be a number of hash
function evaluations per time step. 
We now define \emph{resource allocator}, an abstraction that will allow us to
reason about different resources. 
\begin{definition}[Resource Allocator] 
  A resource allocator, $\allocator$, interacts with the processes through input events (\op{\allocator-commit},
  \op{\allocator-validate}) and output events (\op{\allocator-assign}, \op{\allocator-is-committed}):
  \begin{itemize}
    \item  
    \op{\allocator-commit($p_i, \bstate, r$)}:
    At time step $t$, every process $p_i$ may request a resource
    commitment $\pi$ from the resource allocator by invoking
    $\op{\allocator-commit}$ on inputs a state \bstate and {a resource budget $0 \leq r
    \leq \fun{Alloc}(p_i,t)$, i.e., $p_i$ does not $\op{\allocator-commit}$ more resources than it possesses}. 
    {At the end of {the activation period of~$p_i$}}, the resource allocator either assigns a resource commitment $\pi$ 
    and {a resource budget $r$} to process $p_i$ through an $\op{\allocator-assign}(p_i,\bstate, r, \pi)$ event
    \sloppy{or assigns an empty value $\bot$ and \emph{possibly} a resource $r$ to $p_i$
    through $\op{\allocator-assign}(p_i,\bstate, r, \bot)$. }
    \item
    \op{\allocator-validate($p_i, \bstate,\pi$)}:
    Every process $p_i$ may validate a resource commitment $\pi$ by
    invoking $\op{\allocator-validate}$ on input a state, \bstate, and a resource commitment
    $\pi$. The resource allocator validates the resource commitment $\pi$,
    through an event {$\op{\allocator-is-committed}(p_i, \bstate, b)$} event, with $b=\true$ if
    the commitment $\pi$ is a valid resource commitment for the state \bstate or $b=\false$
    otherwise.
  \end{itemize}
\end{definition}
A process triggers \op{\allocator-commit} to pledge its resources to a system,
and it can be assigned a resource commitment as a result to extend the
blockchain.  
If the resource commitment $\pi$ is included on-chain, then it must be
\emph{valid} (i.e., \op{\allocator-validate} returns $\true$) for the block to
be accepted.
Moreover, we assume that all the events {to and from the resource allocator}
happen within the same time step. 
In particular, if a process $p_i$ $\op{\allocator-commit}$s some resource budget $r$ at time step $t$, at the end of {the activation period for~$p_i$} process $p_i$ will receive either a resource commitment $\pi$ and $r$ or an empty resource commitment value $\bot$ and $r$. 

Resources can be classified into various types. 
In our model, these types can be described as the interactions between processes and the resource allocator.
The following definition classifies different types of resources
used in existing blockchain protocols.

\begin{definition}[Types of resource]
   A resource can be classified as follows.
   \begin{description}[style=unboxed]
       \item[Virtual]
        A resource is \emph{virtual} when the resource allocator determines the
        resource budget of all processes from the given blockchain
        state $\bstate$.  
        For a virtual resource, we assume that there exists a function
        $\fun{StateAlloc}: \mathcal{P} \times \mathbb{C} \rightarrow \BN$ that
        takes as input a process $p_i$ and a blockchain $\chain$ and outputs the resource budget of $p_i$, and $p_i$ can invoke
        $\op{\allocator-commit}(\cdot, \cdot, r)$ on
        an empty resource, $r=\bot$.
       \item[External] 
        A resource is \emph{external} when a process must allocate the resource {\em externally} with a budget $r \geq 0$ to invoke $\op{\allocator-commit}()$. {For an external resource, this commitment step is equivalent to giving \allocator access to the external resource with the budget $r$. Moreover,
        we assume that processes cannot lie about the resource budget $r$ and commit more than $r$.} 
       \item[Burnable] 
        A resource is \emph{burnable} when a process $p$ can trigger 
        \emph{multiple} 
        $\op{\allocator-commit}(\cdot, \cdot, r)$ at a time step $t$, and it 
        retrieves $r$ through $\op{\allocator-assign}(\cdot, \cdot, r, \cdot)$
        at the end of {the activation period for $p_i$}.
        For all committing events $\op{\allocator-commit}(p_i, \cdot, r_i)$ from the same process $p_i$ that occur within a time step $t$ , we require $\sum_{r_i > 0} r_i \leq \fun{Alloc}(p_i, t)$.
       \item[Reusable]
       A resource is \emph{reusable} when a process $p_i$ can use the same
       resource budget $r \leq \fun{Alloc}(p_i,t)$ to trigger \emph{infinitely many}
       $\op{\allocator-commit}(\cdot, \cdot, r)$ at each time step $t$, and $p_i$ 
       does not need to retrieve $r$ from the output event $\op{\allocator-assign}$. Hence, for reusable resources, we denote the value of $r$ in the output event $\op{\allocator-assign}(\cdot, \cdot, r, \cdot)$ to be $\bot$.
   \end{description}  
\end{definition}
\begin{remark}
{The assumption on external resources is natural because an \emph{external}
resource is inherently unforgeable; for instance, in \powork, processes cannot
fake this budget as it is the physical limit of the mining hardware. For
resources like storage, the resource is the physical hard drive, and $r$ can be
thought of as the capacity of the hard drive.}
\end{remark}

\pparagraph{Failures} A process that follows its protocol during an execution
is called \emph{correct}.  On the other hand, a \emph{faulty} process may crash
or deviate arbitrarily from its specification, such processes are also called \emph{Byzantine}.  
We consider only Byzantine faults in this work. 
All Byzantine processes are controlled by a probabilistic polynomial-time adversary, 
\CA; {we write $p_i \in \CA$ to denote that a Byzantine process is controlled by $\CA$}. 
{In this model, we require the adversary to go through the same process of committing resources and getting assigned
resource commitments from the allocator. 
Since the allocator assigns the commitment at the end of the time step, 
we require a minimum delay between Byzantine processes to be one.
We also note that this requirement is only for definitional reasons 
and can be relaxed by assuming the network delay to be zero {for Byzantine processes}. 
However, the concrete parameters on the probability of getting assigned resource commitments for Byzantine processes will need to be adjusted to reflect this assumption, and we leave this to future work.}

\begin{definition}[Adversarial Resource Budget] 
{$R_\CA$ is the maximum adversarial resource budget.  
For any time step $t$, it holds that:
$\sum_{p_i\in \mathcal{A}}\fun{Alloc}(p_i,t) \leq R_\CA$.}
\end{definition}

\begin{definition}[Corruption]~\label{def:corruption}
{At any time step $t$, an adversary $\mathcal{A}$ can allocate a resource budget of $\fun{Alloc}(p_i,t)$ from $R_\CA$ to corrupt a process $p_i \in \CP_{\leq t}$.}
\end{definition}

\section{Resource-based Total-order Broadcast}
\label{sec:nakamoto_consensus}

\begin{figure}[t]
\begin{framed}
\noindent\var{unordered}: set of transaction $\tx$ that has been received for execution and ordering \\
$\var{delivered}:$ set of transaction $\tx$ that has been executed and ordered \\
$k$: common prefix parameter\\ 
$\localBlocks$: set of received blocks, $B = (h, \overline{\tx}, \pi, \sigma)$, initially containing $B_0= (\bot, \overline{\tx}, \bot, \bot)$\\
$\mathbb{C}$: set of valid blockchains derived from $\mathcal{B}$, initially contains one chain $\chain = [B_0]$\\
$\lchain$: local selected blockchain  \\
\btx: a \op{\allocator-commit}ted block for $p_i$\\
{At time step $t$, $r_i = \fun{Alloc}(p_i, t)$ if $r_i$ is \emph{external}, $r_i \leftarrow \bot$ if $r_i$ is \emph{virtual}}
\end{framed}
\caption{Initial state of a correct process}
\label{fig:initial-state}
\end{figure}

In this section, we define an algorithm for the \emph{resource-based
longest-chain total-order broadcast} using a (\emph{probabilistic}) \emph{resource
allocator} $\gR$. {We define various properties needed for a secure resource allocator so that the generic algorithm correctly guarantees properties of total-order broadcast}. 
Then, we concretely define three different resource
allocators based on three types of resource: computation, stake, and storage
to inherently capture three popular ({\em probabilistic}) blockchain protocols,
namely, Nakamoto consensus,
Ouroboros Praos, and Filecoin's consensus protocol.

\subsection{Generic Resource-based Longest-chain Total-order Broadcast}~\label{sub:longest-chain-protocol}

\begin{algo}
   \vbox{
   \small
   \begin{numbertabbing}
     xxxx\=xxxx\=xxxx\=xxxx\=xxxx\=xxxx\=MMMMMMMMMMMMMMMMMMM\=\kill
     \textbf{uses} \label{}\\
     \> Resource allocator: $\gR$ \label{}\\
     \> Probabilistic reliable broadcast: \op{gossip} \label{}\\
     \> Validation predicate: $\PP(\cdot, \cdot)$  \label{} \\
     \> Random oracle: $H:\bset^* \rightarrow \bset^\lambda$ \label{}\\
     \> Signature scheme: $\Sigma = (\sign, \verify)$\label{}\\
 \textbf{init} \label{} \\
    \> $\op{Extend}(\chain=[B_0])$  \label{} \\
    
       \textbf{upon} $\fun{a}$-$\fun{broadcast}(\tx)$ \textbf{do}\label{algo:pox-a-broadcast}\\
     \> \textbf{invoke} $\fun{gossip}$-$\fun{broadcast}(\msg{op}{\tx})$ \label{algo:pox-gossip}\\

     \textbf{upon} $\fun{gossip}$-$\fun{deliver}$ $(\msg{op}{\tx})$ \textbf{do}\label{algo:pox-gossip-deliver}\\
     \> $\var{unordered} \becomes \var{unordered}\cup\set{\tx}$ \label{algo:pox-unordered}\\
     {\textbf{upon} $\fun{gossip}$-$\fun{deliver}$ $(\msg{request}{B})$ \textbf{do}} \` $\rhd$ Receive a request for parents of $B$ \label{line:sync0}\\
     \> {\textbf{if} $\exists~\chain \in \BC$ \textbf{s.t.} $\chain[-1] = B$ \textbf{then}} \label{line:sync1}\\
     \>\>{\textbf{forall} $B'\in \chain$ \textbf{do}} \` $\rhd$ Re-send all parents of $B$\label{line:sync2} \\
     \>\>\> {\textbf{invoke} $\fun{gossip}$-$\fun{broadcast}$ $(\msg{blk}{B'})$} \label{line:sync3} \\
     \textbf{upon} $\fun{gossip}$-$\fun{deliver}$ $(\msg{blk}{B})$ \textbf{s.t.} $B=(h,\overline{\tx}, \pi,$ $\sigma_i$) \textbf{do}\label{algo:pox-gossip-deliver-block}\\
     \> \textbf{if} $\verify(p_j, h||\overline{\tx}||\pi||\slot_j,\sigma_j)$  $\wedge~\exists~\chain \in \mathbb{C}$ \textbf{s.t.} $H(\chain[-1]) = h$ $\wedge$ $\PP(\chain, \overline{\tx})$ \textbf{then} \label{line:validatesx} \\
     \>\> $\bstate \becomes (\chain, B)$ \label{} \\
     \>\> \textbf{invoke} $\op{\allocator-validate}(p_i, \bstate, \pi)$ \label{line:r-validate} \\
     \> {else} \label{}\\
     \>\> {\textbf{invoke} $\op{gossip}\text{-}\op{broadcast}(\msg{request}{B})$} \` $\rhd$ Request for parents of $B$\label{line:orphan} \\
     \textbf{upon} $\op{\allocator-is-committed}(p_i,  \bstate, \true)$ \textbf{s.t.} $\bstate = (\chain, B)$ \textbf{do}  \label{} \\
     \> $\mathcal{B} \becomes \mathcal{B} \cup \set{B}$ \label{}\\
     \> \textbf{if} $|\chain| > |\lchain|$ \textbf{then} \label{algo:pox-longest-chain} \\
     \>\>  $\lchain \becomes \chain$\label{algo:pox-c-greater-local} \` $\rhd$  Update the local blockchain\\
     \>\> $\op{Extend}(\lchain)$  \label{algo:pox-mine} \\
     \textbf{upon} $\op{\allocator-assign}(p_i, \bstate, r, \pi)$ \textbf{s.t.} $\bstate=(\chain,B = (h,\overline{\tx}, \pi,\bot)), \pi \neq \bot$ \textbf{do} \label{} \\
     \> {$\sigma_i  \gets \sign(p_i, h||\overline{\tx}||\pi)$}  \label{} \\
     \> $B \becomes (h, \overline{\tx}, \pi,${$\sigma_i$})  \label{algo:pox-pi-to-block} \\
     \> \textbf{if} $r_i$ is \emph{burnable} \textbf{then} \label{} \\
     \> \> $r_i \gets r_i+r$  \label{} \\
     \> \textbf{invoke} $\fun{gossip}$-$\fun{broadcast}$ $(\msg{blk}{B})$ \label{algo:pox-gossip-block} \\
     \textbf{upon} $\op{\allocator-assign}(p_i, \bstate, r, \pi)$ \textbf{s.t.} $\pi = \bot$ \textbf{do} \label{} \\ 
     \> \textbf{if} $r$ is \emph{burnable} \textbf{then} \label{} \\
     \>\> $r_i \gets r$  \label{} \\
     \> {$\op{Extend}(\lchain)$}  \label{} \\
    \textbf{upon} $\fun{a}$-$\fun{deliver}$ $(\msg{op}{\tx})$ \textbf{do}\label{}\\
     \> $\var{delivered} \becomes \var{delivered}\cup\set{\tx}$ \label{}\\
     \textbf{function} $\op{Extend}(\lchain)$\label{}\\
     \> \textbf{forall} $\tx \in  \lchain[:-k]$ $\wedge$ $\tx\notin \var{delivered}$  \textbf{do} \label{line:delivered} \\
     \>\>  \textbf{output} \fun{a}-\fun{deliver}$(\msg{op}{\tx})$ \` $\rhd$  Deliver all transactions in the common prefix\label{algo:pox-a-deliver} \\
     \>\> {$\var{unordered} \becomes \var{unordered} \setminus \{\tx\}$} \label{} \\
     \> $h \becomes H(\lchain[-1])$ \label{}\\
     \> select a list of transactions $\overline{\tx}$ from \var{unordered} such that $\PP(\chain, \overline{\tx})=\true$ \label{algo:pox-valid-list-bitcoin} \\
     \> $\btx \becomes (h, \overline{\tx}, \bot, \bot)$ \label{line:xinblockPOW}\\
     \> \textbf{invoke}  $\op{\allocator-commit}(p_i, (\lchain, \btx), r_i)$ \label{algo:pox-commit}\\
     \> \textbf{if} $r$ is \emph{burnable} \textbf{then}\label{} \\
     \>\> $r_i \gets 0$  \label{line:resource-burnt} 
   \end{numbertabbing}
   }
   \caption{Resource-based longest-chain total-order broadcast (process $p_i$).}
   \label{alg:tob}
\end{algo} 

A protocol for resource-based longest-chain total-order broadcast using $\gR$
allows any process $p_i$ to \emph{broadcast} transactions by invoking
$\op{a-broadcast}(\tx)$ and to \emph{deliver} those that are valid (according
to a validation predicate $\mathbb{P}(\cdot, \cdot)$ and the local chain,
$\lchain$) through an $\op{a-deliver}(\tx)$ event. Delivered transactions are
\emph{totally ordered} and stored in a list, \var{delivered}, by every process. 

In particular, when a process $p_i$ $\op{a-broadcast}$s a transaction, this is
gossiped to every process, and eventually every correct process
$\op{gossip-delivers}$ it and stores it in a set \var{unordered}. 
Every stored transaction is then considered by $p_i$.

At any given time, a process may receive new blocks from other processes. Any
process $p_i$ can validate the block by invoking $\op{\allocator-validate}$ and
$\op{\allocator-assign}$ed resource commitment to a process $p_j$ by $\gR$.
Once the resource commitment is validated, the process verifies other
components of the block such as signature and transactions and store new blocks
in $\CB$.  {Also, if a block $B$ received by other processes does not have a
parent (L\ref{line:orphan}), the process can trigger a request message to pull
the blockchain $\chain$ with $B$ as the tip from other processes. Upon
receiving this request message, other processes re-broadcast every blocks in
$\chain$ with $B$ as the tip (L\ref{line:sync0}-L\ref{line:sync3}). This step
is an oversimplified and inefficient version of how blockchain nodes
synchronize the chain with others. The goal is to demonstrate that it is
feasible to obtain old blocks from other processes.}

At any given time, a correct process adopts the longest chain to its knowledge
as its local chain $\lchain$, and \op{extends} with a block $\btx$ it wishes to
order at the last block of its local chain $\lchain$. Observe that the
$\op{Extend}$ function in Algorithm~\ref{alg:tob} captures the operation of
creating new blocks, usually called \emph{mining}, in blockchain protocols and
we refer to the processes in charge of creating blocks as \emph{miners} or
\emph{validators}, interchangeably.

In our model, we abstract this \emph{extending} operation as the interaction between the processes and the resource allocator $\gR$. Namely, to start extending, process $p_i$ needs to allocate a resource $r$ along with the proposed state $(\lchain, B_{\text{com}})$ to the resource allocator $\gR$ through $\op{\allocator-commit}()$. 
Once $\gR$ $\op{assign}s$ a resource commitment $\pi$, $p_i$ attaches $\pi$ to the block and gossips the block to other processes. The details of this interaction differ depending on the type of resource and are left for the next subsections.
Figure~\ref{fig:initial-state} specifies all data structures maintained by a process, and the code for a process is presented in Algorithm~\ref{alg:tob}.

{For Algorithm~\ref{alg:tob} to satisfy the properties of total-order broadcast, the generic resource allocator needs to satisfy various properties, and we define those properties as follows.}
\begin{definition}[Secure Resource Allocator] 
\label{def:commitment-properties}
  A resource allocator $R$ is secure if it satisfies the following properties:
  \begin{description}[style=unboxed]
    \item[Liveness] 
        {At a time step $t$, if a process $p_i$ invokes
        $\fun{\allocator-commit}(p_i, \bstate, r)$ with a state $\bstate$ and a
        resource budget $r$ then $R$ issues either $\fun{\allocator-assign}(p_i,
        \bstate, r, \pi)$ or $\fun{\allocator-assign}(p_i, \bstate, r, \bot)$
        during time step $t$.}
    \item[Validity] 
    If resource commitment $\pi$ is a valid resource commitment $(i.e., \pi \neq \bot)$
    contained in an output event $\op{\allocator-assign}(p_i,\bstate,r, \pi)$, 
    then any process $p_j$ can invoke
    $\op{\allocator-validate}(p_j,\bstate,\pi)$. The resource allocator $R$
    outputs $\op{\allocator-is-committed}(p_j, \bstate, b)$ with $b=\true$.
    \item[Use-Once] 
    At any time step $t$, for any states $\bstate$, $\bstate_1$, $\bstate_2$, 
    any resource budget $r$, $r_1, r_2 \in \BN$ such that $r_1+r_2 = r$,   
    the probability that \allocator responds with 
    \op{\allocator-assign}$(p_i,\bstate,r,\pi)$ with $\pi\neq \bot$ after
    $\op{\allocator-commit}(p_i, \bstate, r)$
    is \emph{greater or equal} to the probability 
    that \allocator responds \emph{at least one}
    \op{\allocator-assign}$(p_i,\bstate_i,r_i,\pi)$ for $i\in\{1,2\}$ with $\pi
    \neq \bot$ 
    after \emph{two} $\op{\allocator-commit}(p_i, \bstate_1, r_1)$ and
    $\op{\allocator-commit}(p_i, \bstate_2, r_2)$. 

    {For a \emph{reusable resource}, at any time step $t$, a resource budget $r$,
    a state $\bstate$ and upon potentially multiple repeated
    $\op{\allocator-commit}(p_i, \bstate, r)$ from the same process $p_i$,
    if \allocator responds with $\op{\allocator-assigns}(p_i, \bstate, r,
    \pi)$, then $\pi$ is the same for every $\op{\allocator-commit}$ events
    output by $\allocator$.}
  \item[Unforgeability] No adversary can produce a resource commitment $\pi$
    such that $\pi$ has not been previously $\op{\allocator-assigned}$ by
    $\allocator$ and, upon \op{\allocator-validate}($p_i,$ $\bstate,\pi$),
    \allocator triggers \op{\allocator-is-committed}($p_i, \bstate,\true)$, for
   some state \bstate and some process~$p_i$.
   \item[Honest-Majority Assignment] At each time step, we denote with $\hprob$ and $\aprob$ the probabilities that at least one correct process and one Byzantine process, respectively, obtain a valid resource commitment for each $\op{\allocator-commit}$. More formally, for every time step $t$, we define: 
   \begin{equation*}
     \begin{split}
        \aprob &= \Pr[\exists \op{\allocator-assign}(p_i, \bstate, r, \pi)~\text{such that}~\pi\neq \bot \wedge p_i \in \CA],\\
        \hprob &= \Pr[\exists \op{\allocator-assign}(p_i, \bstate, r, \pi)~\text{such that}~\pi\neq \bot \wedge p_i \notin \CA].
     \end{split}
   \end{equation*}
   Then we require that:
   \begin{equation}~\label{eq:threshold}
        \aprob < \frac{1}{\Delta-1+1/\hprob}.
   \end{equation}
  \end{description}
\end{definition}

{The \emph{liveness} property  aims
to capture the mining process in permissionless \powork blockchains and ensure that if
processes keep committing resources, \emph{eventually} one process will get
assigned the resource commitment to extend the blockchain.}

{The \emph{validity} property guarantees that a resource commitment can always
be verified by any process $p_i$ by triggering at any point
$\op{\allocator-validate}$. 
This property captures the fact that any participant can efficiently verify, for example, the validity of the solution to the
computational puzzle in \powork protocols or the evaluation of the verifiable
random function in \postake protocols}.

{The \emph{use-once} property prevents processes from increasing the
probability of getting assigned the resource commitment either by committing
several times, splitting the resource budget and then committing all the
split amounts at different states or by committing a smaller resource budget.
Intuitively, the \emph{use-once} property also implies that the property holds for any
integer partition of $r$ (i.e., $r = \sum_{r_i > 0}{r_i}$).
Moreover, the use-once property also implies that our model mainly
focuses on probabilistic protocols as we do not aim to bypass the lower
bound established in~\cite{DBLP:journals/corr/abs-2101-07095}, namely, there is
no deterministic protocol in \emph{permissionless} setting {that solves
consensus}. On the other hand, we believe that applying our model to
\emph{permissioned} blockchains with \postake, e.g.,
Tendermint~\cite{DBLP:journals/corr/abs-1807-04938}, can be interesting future work. }

The \emph{unforgeability} property ensures that no process $p_i$ can
produce a valid resource commitment $\pi$ that has not been previously
$\op{\allocator-assigned}$ by the resource allocator.

{Finally, the \emph{honest-majority assignment} implies that despite the network delay, correct processes will have a higher probability of getting assigned the resource commitment at each time step. 
\cref{eq:threshold} was established by Ga\v{z}i~\etal~\cite{DBLP:conf/ccs/GaziKR20}, and it takes into account that honest blocks may get discarded due to the network delay $\Delta$.}

\pparagraph{Security Analysis} 
With the defined properties of a secure allocator, 
our model is equivalent to the idealized model
introduced by Ga\v{z}i~\etal~\cite{DBLP:conf/ccs/GaziKR20}. Therefore, their result
also holds for our protocol, and we present them in our model as follows.

\begin{lemma}[\cite{DBLP:conf/ccs/GaziKR20}]~\label{lemma:liveness}
  {Algorithm 1 implemented with a secure resource allocator \gR satisfies the following properties:}
  \begin{description}[style=unboxed]
    \item[Safety] {For any time steps $t_1$ and $t_2$ with $t_1 \leq t_2$, 
      a common prefix parameter $k$ and 
      any local chain maintained by a correct process $\lchain$,
      it holds that $\lchain^{t_1}[:-k]\preceq \lchain^{t_2}$ with an overwhelming
      probability.}
    \item[Liveness] {For a parameter $u$ and any time step $t$, 
      let $\lchain$ be the local chain maintained by a correct process, then
      there is at least one
      new honest block in $\chain^{t+u}/\chain^{t}$ with an overwhelming probability.}
  \end{description}
\end{lemma}

{Intuitively, \emph{safety} implies that correct processes do not deliver
different blocks at the same height, while
\emph{liveness} implies that every transaction is eventually delivered by all
correct processes.
Using \cref{lemma:liveness}} and properties of a secure resource allocator, 
we conclude the following.

\begin{theorem}
\label{thm:tob}
If $\gR$ is a \emph{secure} resource allocator,
then Algorithm~\ref{alg:tob} implements total-order broadcast.
\end{theorem}
\begin{proof}
{Observe that, since \gR is a secure resource allocator, it satisfies \emph{use-once} property.
Therefore, Byzantine processes 
cannot amplify the probability $\aprob$ by repeatedly triggering \op{\allocator-commit()} on reusable resources at the same time step.}

For the \emph{validity} property, if a correct process $p_i$ \op{a-broadcasts} a transaction $\tx$ (L\ref{algo:pox-a-broadcast}), $\tx$ is \op{gossip-broadcast} (L\ref{algo:pox-gossip}) and, after $\Delta$, every correct process $\op{gossip-delivers}$ $\tx$ (L\ref{algo:pox-gossip-deliver}) and adds it to \var{unordered} (L\ref{algo:pox-unordered}).
Eventually, transaction $\tx$ is selected by a correct process $p_j$ as part of a block $B$ (L\ref{algo:pox-valid-list-bitcoin}). Block $B$ is then \op{gossip-broadcast} by $p_j$ (L\ref{algo:pox-gossip-block}) and eventually, after $\Delta$, every correct process $\op{gossip-delivers}$ $B$ (L\ref{algo:pox-gossip-deliver-block}), validates $x$ (L\ref{line:validatesx}), and validates the resource commitment (L\ref{line:r-validate}).
{Observe that, because of the \emph{unforgeability} property of $\gR$, a valid resource commitment cannot be produced except by the resource allocator.}
Observe that this last step is possible through the validity property of $\gR$.
The proof then follows from Lemma~\ref{lemma:liveness}.

\emph{No duplication} property follows from the algorithm; if a correct process $p_i$ $\op{a-delivers}$ a transaction $\tx$, $p_i$ adds $\tx$ to $\var{delivered}$ and condition at line L\ref{line:delivered} cannot be satisfied again.

For the \emph{agreement} property, let us assume that a correct process $p_i$ \op{a-delivered} a transaction $\tx$ buried at least $k$ blocks deep in its adopted chain $\chain$. 
Process $p_i$ \op{a-delivers} a transaction $\tx$ when it \emph{updates} the local blockchain with the longest chain $\chain$ (L\ref{algo:pox-c-greater-local}), $\tx$ has not been $\op{a-delivered}$ yet and $\tx$ is part of the common prefix $\chain[:-k]$ (L\ref{line:delivered}).
The property then follows from Lemma~\ref{lemma:liveness}; eventually every correct process $\op{a-delivers}$ transaction $x$, with an overwhelming probability.

Finally, for the \emph{total order} property, 
{from the safety property of Lemma~\ref{lemma:liveness}, 
we know that correct processes do not deliver different blocks at the same height. 
This means that at a given height, if two correct processes $p_i$ and $p_j$ \op{a-delivered} a block, then this block is the same for $p_i$ and $p_j$ with an overwhelming probability.
Moreover, since a block is identified by its hash, due to the collision-resistance property of $H(\cdot)$, it also implies that the set and order of transactions included in the block are the same for every correct process.}
So, if process $p_i$ \op{a-delivers} transaction $\tx_1$ before $\tx_2$, then either $\tx_1$ and $\tx_2$ are in the same block $B$ with $\tx_1$ appearing before $\tx_2$ or they are in different blocks $B_1$ and $B_2$ such that $B_2$ appears in the chain after $B_1$.
The total-order property follows.
\end{proof}

\subsection{Proof-of-Work Resource Allocator}%
\label{sub:proof_of_work}

In this part, we present the \powork allocator as a concrete instantiation of
the resource allocator for \emph{burnable} and \emph{external}
resources. 

\pparagraph{Proof-of-Work Resource Allocator} The \powork resource allocator $\powR$ is parameterized by $\dprob$ which is the default probability of getting assigned resource commitment for $r=1$. $\powR$ works as follows. Upon $\op{\allocator-commit}(p_i,\bstate,r)$ by process $p_i$ with a valid chain $\chain$ {with respect to $B_0$}, $\powR$ starts $r$ concurrent threads of $\op{Pow}()$ function which acts as a biased coin with probability $\dprob$ of assigning the resource commitment. {Observe that, because computation is a \emph{burnable} and \emph{external} resource, processes cannot lie on about the $\op{committed}$ resource budget $r$.} In particular, $\op{Pow}$ uniformly sample a value $\var{nonce}$ in $\bset^{\lambda}$ and either returns $\var{nonce} \in \bset^{\lambda}$ or $\bot$. If $\var{nonce}$ is the returned value in $\bset^\lambda$, then $\powR$ \op{assigns} it as the resource commitment to $p_i$, otherwise it $\op{\allocator-assigns}$ $\bot$ to $p_i$. If the committed chain $\chain$ is not valid, then $\powR$ $\op{\allocator-assigns}$ $\bot$ to $p_i$. 
Validation of the resource commitment can be done by any process $p_j$ through
$\op{\allocator-validate}$; the resource allocator $\powR$ returns either \true
or \false, depending on the validity of the resource commitment. We implement
the resource allocator $\powR$ in Algorithm~\ref{algo:pow} and obtain the 
following lemma and theorem.

\begin{algo}
   \vbox{
   \small
   \begin{numbertabbing}
    xxxx\=xxxx\=xxxx\=xxxx\=xxxx\=xxxx\=MMMMMMMMMMMMMMMMMMM\=\kill
    \textbf{state} \label{}\\
    \> {$B_0$: Genesis block} \label{}\\
    \> $\dprob$: Default probability of getting assigned resource commitment on \emph{one} resource \label{}\\
    \textbf{uses} \label{}\\
    \> Random oracle: $H:\bset^* \rightarrow \bset^\lambda$ \label{}\\
    \textbf{upon} $\op{\allocator-commit}(p_i,\bstate, r)$ \textbf{s.t.} $\bstate = (\chain, B), B=(h, \overline{\tx}, \bot, \bot)$ \textbf{do} \label{}\\ 
    \> \textbf{if} $\chain$ is valid $\wedge~H(\chain[-1])=h$ \textbf{then} \label{line:C-valid} \\
    \>\> \textbf{start} $r$ concurrent threads with $\var{nonce}_j = \op{Pow}(p_i,\bstate)$ for $j\in \{0, \dots, r-1\}$  \label{} \\
    \>\> {\textbf{wait} for all $r$ threads with $\op{Pow}(p_i,\bstate)$ for $j\in \{0, \dots, r-1\}$ to  finish}\label{}\\
    \>\> \textbf{if} $\exists~\var{nonce}_j \neq \bot$ \textbf{then} \label{}\\
    \>\>\> \textbf{output} $\op{\allocator-assign}(p_i, \bstate, r, nonce_j)$ \label{}\\
    \>\> \textbf{else} \label{line:C-invalid-1}\\
    \>\>\> \textbf{output} $\op{\allocator-assign}(p_i, \bstate, r, \bot)$ \label{}\\
    \>  \textbf{else} \label{line:C-invalid-2} \\
    \>\>\textbf{output} $\op{\allocator-assign}(p_i, \bstate, r, \bot)$ \label{}\\
    \textbf{upon}  $\op{\allocator-validate}(p_i, \bstate, \pi)$ \textbf{s.t.} $\bstate=(\chain, B), B=(h,\overline{\tx},\pi, \sigma), \pi=\var{nonce}$ \textbf{do} \label{}\\
    \> $\var{b} \gets H(h||\overline{x}||\var{nonce}) \stackrel{?}{\leq} \dprob\cdot 2^{\lambda}$  $\wedge~\chain$ is valid $\wedge~H(\chain[-1])=h$\label{} \\ 
    \> \textbf{output} \op{\allocator-is-committed}$(p_i, \bstate, \var{b})$\label{}\\
    \textbf{function} $\op{Pow}(p_i,\bstate)$ \` $\rhd$ With $\bstate = (\chain, B)$ and $B=(h, \overline{x}, \bot)$\label{}\\
    \> $\var{nonce} \sample \{0,1\}^\lambda$ \label{}\\
    \> \textbf{if} $H(h||\overline{x}||\var{nonce}) \leq \dprob\cdot 2^{\lambda}$ \textbf{then}\label{linne:foundt}\\
    \>\>\textbf{return} $\var{nonce}$ \label{}\\
    \> \textbf{return} $\bot$ \label{}
    \caption{Implementing PoW resource allocator, $\powR$}
    \label{algo:pow}
   \end{numbertabbing}
   }
\end{algo}

\begin{lemma} 
\label{lem:R_pow}
  {Given the random oracle $H(\cdot)$, 
  the default probability $\dprob$ of getting assigned resource commitment on $r=1$, and 
  the network delay $\Delta$,
  there exists a value $R_\CA$ such that the resource allocator $\powR$ implemented in Algorithm~\ref{algo:pow} is a secure resource
  allocator.}
  \end{lemma}
\begin{proof} 
  \emph{Liveness} property follows from Algorithm~\ref{algo:pow}: upon
  $\op{\allocator-commit}(p_i, \var{st}, r)$ from process $p_i$, the resource allocator
  $\powR$ either \emph{(i)} has L\ref{line:C-valid} satisfied and, eventually,
  outputs $\op{\allocator-assign}(p_i, \var{st}, r, \pi)$ with a resource commitment $\pi$
  to $p_i$ or outputs $\op{\allocator-assign}(p_i, \var{st}, r, \bot)$ to $p_i$ or \emph{(ii)} if the chain is invalid (L\ref{line:C-valid}), and then the allocator outputs $\op{\allocator-assign}(p_i, \var{st},
  r, \bot)$ to $p_i$.

  For the \emph{validity} property, observe that $\pi=\var{nonce}$ is valid if and only if 
  there is a valid chain \chain with respect to the genesis block $B_0$ such
  that the last block $B=\chain[-1]$ contains $\pi$. Hence,
  $\pi$ can be validated by any process $p_j$ through $\op{\allocator-validate}(p_j,
  (\chain, B^*), \pi)$; $\powR$ then checks if $H(B^*)=h$
  and $H(h||\overline{x}||\var{nonce}) \leq \dprob \cdot 2^{\lambda}$ outputting the same result at any process $p_j$.

  The \emph{use-once} property immediately follows because of the burnable property of the underlying resource (i.e., computation). 
  Since $\powR$ triggers $\op{\allocator-assign}$s at the end of the  {the activation period for~$p_i$},
  we claim that for multiple $\op{\allocator-commit}(p_i, \cdot, r_i)$ committed by $p_i$ at $t$, it is equivalent to trigger $\op{\allocator-commit}(p_i, \cdot, r)$ once for $r = \fun{Alloc}(p_i,t)$. In particular, at the time step $t$, 
  let $\mathsf{Bad}$ be the event of not getting any resource commitment on all
  $\op{\allocator-commit}(p_i,\cdot, r_i)$ for $r_i\in \{r_1, r_2\}$ such that $r = r_1+r_2$,
  then the probability 
  of getting the resource commitment is  $1-\Pr[\mathsf{Bad}] = 1 -(1-1+(1-\dprob)^{r_1})(1-1+(1-\dprob)^{r_2}) = 1-(1-\dprob)^{r}$.

  Since the resource allocator $\powR$ uses the random oracle $H$ (i.e., idealized hash function with no exploitable weaknesses),
  the \emph{unforgeability} property follows from the observation that, 
  in order to produce a valid resource commitment, $p_i$ has no better way to find the solution
  than trying many different queries to $H$. 
  This implies that $p_i$ has the same probability of obtaining a valid local resource commitment as it would have by $\op{\allocator-committing}$ to the resource allocator.

  {For the \emph{honest-majority assignment} we recall that $\hprob$ and $\aprob$ are the probabilities that at least one
  correct process and one Byzantine process get the resource commitment after
  one $\op{\allocator-commit}$, respectively.
  {In our model, it is not difficult to see that  
  $\hprob = 1-(1-\dprob)^{R-R_\CA}$ and $\aprob = 1-(1-\dprob)^{R_\CA}$. 
  Therefore, one can easily derive the amount of resource $R_\CA$ such that $\aprob < \frac{1}{\Delta-1+1/\hprob}$.}}
\end{proof}

 \begin{theorem}
 \label{thm:powtob}
   {Algorithm~\ref{alg:tob} with the secure resource allocator $\powR$ implements total-order broadcast.}
 \end{theorem}
 \begin{proof}
   From Theorem~\ref{thm:tob} and Lemma~\ref{lem:R_pow}, it follows that, since $\powR$ is a secure resource allocator, then Algorithm~\ref{alg:tob} with $\powR$ implements total-order broadcast.
 \end{proof}

\subsection{\PoStake Resource Allocator}
\label{sub:pos-allocator}

The resource in \postake protocols is the stake of each process, and stake is a \emph{virtual} and \emph{reusable} resource.
In those protocols, the probability of a process $p_i$ being 
assigned a resource commitment is proportional to its stake in the system. 
In this part, we focus on Ouroboros
Praos~\cite{DBLP:conf/eurocrypt/DavidGKR18} for our formalization.
Before presenting the \postake resource allocator as a concrete instance of a resource allocator for \emph{reusable} and \emph{virtual} resources, we need to introduce additional considerations and definitions.

\pparagraph{Reusable Resources}
By definition, a reusable resource allows processes to repeatedly trigger $\op{\allocator-commit}$ 
in the same time step using the same resource. 
Hence, if \allocator does not satisfy \emph{use-once} property and assigns a resource commitment randomly, then every time a process $p_i$ triggers $\op{\allocator-commit}$, process $p_i$ might end up with a different result.
For this reason, a na\"{i}ve implementation of the resource allocator for reusable resources would allow an adversary to amplify the probability
of {getting assigned a resource commitment} (i.e., $\aprob$) by repeatedly invoking \op{\allocator-commit} 
using the same resource, i.e., in a grinding attack~\cite{DBLP:journals/iacr/BonneauCG15} at the same time step.
{Hence, for \emph{probabilistic} blockchain
protocols, to cope with this problem, one needs to ensure that the allocator satisfies \emph{use-once} property.}
In particular, 
from designs of \postake protocols like Snow White~\cite{DBLP:conf/fc/DaianPS19} and Ouroboros Praos~\cite{DBLP:conf/eurocrypt/DavidGKR18}, three commonly used approaches  to enforce \emph{use-once} property are:

\begin{description}[style=unboxed]
  \item[$(R1)$ Explicit Time Slots]
   The first mechanism to enforcing \emph{use-once} property is to index the resource by time
      slots. Protocols like
      Ouroboros Praos~\cite{DBLP:conf/eurocrypt/DavidGKR18} and Snow White~\cite{DBLP:conf/fc/DaianPS19} 
      require processes to have synchronized
      clocks to explicitly track time slots and epochs to ensure that
      each process derives a deterministic leader selection result from the same state.
  \item[$(R2)$ Leader Selection from the Common Prefix]
    {This mechanism requires correct processes to
    extract the set of potential leaders from the common prefix. 
    In particular, the common prefix is a shortened local longest chain
    that 
    is with overwhelming probability the same for all correct processes.
    This approach allows them to
    share the same view of potential leaders.}
    
  \item[$(R3)$ Deterministic and Trustworthy Source of Randomness] 
    {The source of randomness has to be trustworthy to ensure
    a fair leader election and to defend against an adaptive adversary that
    might corrupt processes predicted to be leaders for the upcoming time
    slots.
    In addition, the source of randomness has to be deterministic for each time slot 
    and chain state in order to prevent the previously mentioned grinding attack.
    Hence, popular Proof-of-Stake blockchain protocols often rely on
    sophisticated protocols to produce randomness securely.}
 
\end{description}

{\pparagraph{Slot and Epoch} An \emph{epoch} $e$ is a set of $q$ adjacent time slot
$S=\set{\slot_0,\dots, \slot_{q-1}}$. In practice, slot $\slot$ consists of a sufficient number of time steps so that discrepancies between processes' clocks are insignificant, and processes advance the slot at the same speed. {In our model, we simplify this bookkeeping by requiring the allocator to maintain the slot and epoch.}}

\pparagraph{Stake Distribution} 
The \emph{stake distribution} at a time step $t$ is $\stakeDis^t=\{(p_1, r_1), \dots\}$
with $r_i \geq 0$, specifies the amount of stake owned by each process $p_i \in \mathcal{P}_{\leq t}$. 
We denote $\stakeDis^e$ the stake distribution at the beginning of
epoch $e$.  
{The stake distribution $\stakeDis^e$  can be obtained
  from $\fun{StateAlloc}(\chain[0:\slot], p_i)$ for $\slot \leq e\cdot q$ for each $p_i \in \mathcal{P}_{\leq t}$.}

\begin{definition}[Leader Selection Process]~\label{def:posdef}
A \emph{leader selection process} $(\mathcal{D},F)$ with respect to a {stake} distribution
$\stakeDis= \{(p_1,r_1), \dots\}$ is a pair
consisting of a distribution $\mathcal{D}$ and a deterministic function
${F}$. When $\rho \sample \mathcal{D}$, for all  $sl \in \BN$, $F(\stakeDis, sl; \rho)$ outputs process $p_i$ with probability
$ 1 - (1-\dprob)^{r_i} $
where $\dprob$ is the probability of assigning resource commitment for $r=1$ for a given slot.
\end{definition}

\begin{algo}
   \vbox{
   \small
   \begin{numbertabbing}
    xxxx\=xxxx\=xxxx\=xxxx\=xxxx\=xxxx\=MMMMMMMMMMMMMMMMMMM\=\kill
    \textbf{state} \label{}\\
     \> {$B_0$: Genesis block} \label{}\label{}\\
     \> $\mathcal{D}$ : Distribution \label{} \\ 
     \> $F$ : Leader selection function \label{} \\
     \> $\slot$ : Current slot, initially $\slot=0$  \label{} \\
     \> $e$ : Current epoch, initially $e=0$  \label{} \\
     \> $k$ : common prefix parameter  \label{} \\
     \> $q$ : number of slots in an epoch, initially $q=16\cdot k$ \label{}\\
     \> $T$ : set of assigned resource commitments, initially empty  \label{} \\
     \textbf{uses} \label{}\\
     \> Random oracle: $H:\bset^* \rightarrow \bset^\lambda$ \label{}\\
     \textbf{upon} $\op{\allocator-commit}(p_i, \bstate, r_i)$ \textbf{s.t.} $\bstate = (\chain, B)$, $B=(h, \overline{\tx}, \bot, \bot)$ \textbf{do} \label{}\\ 
     \> \textbf{if} $\chain$ is valid $\wedge~H(\chain[-1])=h$ \textbf{then}  \label{line:pos-c-valid} \\ 
     \>\> obtain $\prechain$ by pruning all blocks with slot $> (e-2)\cdot q$, from \chain \label{}\\
     \>\> {\textbf{if} $\prechain = \emptyset$ \textbf{do}}\label{}\\
     \>\>\> {$\prechain \leftarrow [B_0]$}\label{}\\
     \>\> \textbf{if} $\exists (p_i, \prechain, \rho^*, \slot) \in T$ \textbf{then} \` $\rhd$ Queried before \label{}\\
     \>\>\> $\rho \leftarrow \rho^*$ \label{} \\
     \>\> \textbf{else} \label{}\\
     \>\>\> $\rho \sample \mathcal{D}$ \` $\rhd$ Sample a fresh randomness  \label{line:pos-randomness} \\ 
     \>\>\> $T \gets T \cup \{(p_i, \prechain, \rho, \slot)\}$\` $\rhd$ Update $T$ \label{line:T} \\ 
     \>\> obtain the stake distribution $\stakeDis^{e-2}$ from $\prechain$ \` $\rhd$ Evaluate $\fun{StateAlloc}(\cdot, \cdot)$  \label{} \\
     \>\> $p_j\becomes F(\stakeDis^{e-2}, \slot; \rho)$  \label{line:Foutputsp} \\  
     \>\> \textbf{if} $p_i = p_j$ \textbf{then} \label{}\\
     \>\>\> $\pi \gets (p_i, \rho, \slot)$  \label{} \\
     \>\>\> \textbf{output} $\op{\allocator-assign}(p_i, \bstate, \bot, \pi)$ \label{}\\
     \>\> \textbf{else}  \label{line:pos-c-invalid-1}\\
     \>\>\> \textbf{output} $\op{\allocator-assign}(p_i, \bstate, \bot, \bot)$ \label{}\\
     \> \textbf{else}  \label{line:pos-c-invalid-2} \\
     \>\> \textbf{output} $\op{\allocator-assign}(p_i, \bstate, \bot, \bot)$ \label{}\\
     \textbf{upon}  $\op{\allocator-validate}(p_i, \bstate, \pi)$ \textbf{s.t.} $\bstate= (\chain,B), B=(h,\overline{\tx},\pi, \sigma), \pi=(p_j, \rho, \slot)$ \textbf{do}\label{}\\
     \> obtain $\prechain$ by pruning all blocks with slot $> (e-2)\cdot q$, from \chain \label{}\\
     \> {\textbf{if} $\prechain = \emptyset$ \textbf{do}}\label{}\\
     \>\> {$\prechain \leftarrow [B_0]$}\label{}\\
     \> \textbf{if} $\exists (p_i, \prechain, \rho, \slot) \in T$ \textbf{then} \` $\rhd$ Queried before \label{}\\
     \>\> obtain the stake distribution $\stakeDis^{e-2}$ from $\prechain$  \label{} \\
     \>\> $p_j^*\becomes F(\stakeDis^{e-2},\slot; \rho)$  \label{} \\  
     \>\> $\var{b} \gets p_i \stackrel{?}{=} p_j^*$ $\wedge~\chain$ is valid $\wedge~H(\chain[-1])=h$ \label{} \\
     \>\> \textbf{output} \op{\allocator-is-committed}$(p_i, \bstate, \var{b})$  \label{}\\
     \> \textbf{else}  \label{} \\
     \>\> \textbf{output} \op{\allocator-is-committed}$(p_i, \bstate, \false)$  \label{}\\
     \textbf{upon} \op{Timeout} \textbf{do} \` $\rhd$ Increment slot \label{} \\
     \> $\slot \gets \slot + 1$ \label{} \\
     \> \textbf{if} $\slot \mod q = 0$ \textbf{then} \` $\rhd$  Increment epoch \label{}\\
     \>\> $e \gets e+1$  \label{} \\
     \> \op{starttimer}()  \label{}
    \caption{Implementing PoS Resource Allocator, $\stakeR$}
    \label{algo:ouroboros-pos}
   \end{numbertabbing}
   }
\end{algo}

\pparagraph{\PoStake Resource Allocator} 
{The \PoStake resource allocator $\allocator_{\mathsf{pos}}$ with the leader selection process $(\CD, F)$ works as follows. 
First, we require that \stakeR keeps track of the current epoch and time slot to correctly assign the resource commitment to process $p_i$ for the
current slot. \stakeR keeps track of the slot through
$\var{Timeout}$ triggered by the $\op{starttimer}()$ event. 
This approach is to enforce the first requirement of explicit time slots (R1). 
Secondly, upon $\op{\allocator-commit}(p_i,\bstate,r)$ by process $p_i$ with a valid state
$\bstate$ in slot $\slot$, {\stakeR first checks if 
{a random value $\rho \in \CD$ for ($p_i$,\bstate,\slot) has been previously sampled; if so, 
then \stakeR picks it; otherwise, a fresh random value is sampled. 
This requirement ensures a deterministic and trustworthy source of randomness (R3).}
Then, \stakeR obtains the stake distribution of two epochs before,
$\mathbb{S}^{e-2}_{\text{stake}}$ from $\bstate$. {This ensures a leader selection from the common prefix (R2).} 
The resource allocator \stakeR uses $\mathbb{S}^{e-2}_{\text{stake}}$ together with the sampled randomness
as input to the leader selection function $F$ to check if $p_i$ is selected for
the slot $\slot$.}
If this is the case, then $\stakeR$ $\op{assigns}$
the resource commitment to $p_i$, otherwise it $\op{assigns}$ $\bot$. If the
committed chain $\chain$ is not valid, then $\stakeR$
$\op{assigns}$ $\bot$ to $p_i$.} 
{A validation of the resource commitment can be done by any process $p_j$
through $\op{\allocator-validate}$; the resource allocator $\stakeR$ returns
either \true or \false, depending on the validity of the resource commitment.}
Finally, the \postake resource allocator is presented in Algorithm~\ref{algo:ouroboros-pos},
and we conclude the following lemma and theorem.

\begin{remark}
In practice, the randomness generation can be instantiated using verifiable random
function~\cite{DBLP:conf/pkc/DodisY05},  
multiparty coin-tossing~\cite{DBLP:conf/crypto/KiayiasRDO17} protocol, or a random beacon~\cite{DrandDocumentation}. However,
Algorithm~\ref{algo:ouroboros-pos} aims to show the distinction between
\emph{external} and \emph{virtual} resources. 
\end{remark}

\begin{lemma}
\label{lem:R_pos}
  {Given the random oracle $H(\cdot)$, 
  the leader selection process (\CD, F) parameterized by the default probability $\dprob$, and
  the network delay $\Delta$,
  there exists a value $R_\CA$
  such that the resource allocator $\stakeR$ implemented in Algorithm~\ref{algo:ouroboros-pos} is a secure resource allocator.}
\end{lemma}
\begin{proof}
\emph{Liveness} property follows from the algorithm: upon $\op{\allocator-commit}(p_i, \var{st}, r)$ from process $p_i$, the resource allocator $\stakeR$ either \emph{(i)} has L\ref{line:pos-c-valid} satisfied and, eventually, outputs $\op{\allocator-assign}(p_i, \var{st}, \bot, \pi)$ with a resource commitment $\pi$ to $p_i$, or outputs $\op{\allocator-assign}(p_i, \bstate, \bot, \bot)$ to $p_i$ or \emph{(ii)} if the chain invalid (L\ref{line:pos-c-valid}), the allocator outputs $\op{\allocator-assign}(p_i, \var{st}, \bot, \bot)$ to $p_i$. 

For the \emph{validity} property, observe that $\pi=(p_i, \rho, \slot)$ is valid if and only if $p_i$ is a leader for \slot.
If $p_i$ is a leader for \slot, then {in $T$ there must be the random value $\rho$ previously sampled for $p_i$~(L\ref{line:T}). This means that $F$ evaluated on $\rho$ will output again $p_i$.}
Hence, $\pi$ can be validated by any process $p_j$ through $\op{\allocator-validate}(p_j,\var{st}, \pi)$; $\stakeR$ checks if $\pi \in T$ outputting the same result to $p_j$.

{The \emph{Use-once} property follows because, in our model, \stakeR keeps track of previous $\op{\allocator-commit}$ from $p_i$ along with the time slots and states. 
Moreover, the choice of $\mathsf{prob}_i$ is stake-invariant, and it ensures that an adversary cannot increase its probability of being elected leader by dividing its stake into multiple identities. The proof for this is identical to the proof in~\cref{lem:R_pow}.
In practice, this property is enforced by the deterministic outputs of VRF and Hash function along with slot number and the common chain prefix as input.}

{
The \emph{unforgeability} property follows from the fact that any resource commitment produced by $\stakeR$ is stored by the resource allocator in a set $T$ of assigned resource commitments (L\ref{line:T}).
Hence, it is not possible for any process $p_i$ to produce a valid resource commitment that is not in $T$.
In practice, this property is guaranteed by the uniqueness property of verifiable random functions 
or the collision-resistant property of hash functions.
}

For the \emph{honest-majority assignment} property, it is not difficult to see that 
we can derive $\hprob$ and $\aprob$ from $R$ and $R_\CA$. 
In particular, $\hprob = 1-(1-\dprob)^{R-R_{A}}$ and $\aprob \approx 1-(1-\dprob)^{R_A}$. 
{Here we note that the adversary can slightly increase $\aprob$ by committing to shorter chains. 
However, it also means that the adversary will fall behind as it has to extend a much shorter chain than the current local chain maintained by correct processes, and we assume the adversary has no reason to do so. 
Hence, we consider $\aprob = 1-(1-\dprob)^{R_A}$.}
Thus, we can derive $R_A$ so that $\aprob < \frac{1}{\Delta - 1 + 1/\hprob}$.
\end{proof}

\begin{theorem}
\label{thm:poxtob}  
 Algorithm~\ref{alg:tob} with the secure resource allocator $\stakeR$
 implements total-order broadcast.
\end{theorem}
\begin{proof}
 From Theorem~\ref{thm:tob} and Lemma~\ref{lem:R_pos}, it follows that, since $\stakeR$ is a secure resource allocator, then Algorithm~\ref{alg:tob} with $\stakeR$ implements total-order broadcast.
\end{proof}

\section{Proof-of-Space Resource Allocator}%
\label{sub:proof_of_spacetime}
In this subsection, we consider {a \postorage allocator inspired} by Filecoin~\cite{filecoin2017}, 
a blockchain protocol based on storage. {Storage is an example of \emph{external} 
and \emph{reusable} resources.}

\pparagraph{Overview of Filecoin} In the Filecoin protocol, participants get elected to write 
new blocks proportionally to the amount of storage they pledge to the blockchain.
Since storage can be reused,  an adversary could pledge multiple time the same storage.
Hence, Filecoin protocol asks that participants periodically post cryptographic proofs 
of their storage to the blockchain, to ensure that the resources that they pledge are still committed. 
If a miner fails to prove their storage frequently, they are removed from the \emph{power table}, i.e., the mapping from miners' identity to their amount of storage, and lose their right to create blocks.
The maintenance of the power table thus happens on top of the consensus protocol. 
The consensus protocol then uses the power table to elect leaders, just like a \postake blockchain,
except that participants are weighted according to their storage instead of coins.

Filecoin uses a tipset model instead of the longest chain model, meaning that a block can have more than one parent. 
This model is inspired by blockDAGs~\cite{DBLP:conf/podc/KeidarKNS21}, 
although in the tipset model all the blocks included in a tipset must have the same set of parents (unlike blockDAGs).
In this paper however, we are interested in modeling resource allocators based on different resources and regardless of the underlying consensus protocol. For this reason, we choose to use a longest-chain protocol to model Filecoin as this ensures that we can meaningfully compare all the protocols based on their underlying resource rather than their fork-choice rule or consensus protocol. 

{In the following we discuss how to build the \postorage resource allocator based on Filecoin protocol in our model.}

{\pparagraph{Resource budget}
A \emph{sector} is the smallest unit of
storage that a process can commit.
In our model, we define the resource budget $r$ controlled by each process 
to be the number of sectors each process owns.}

{Filecoin implements the three approaches described in~\Cref{sub:pos-allocator} to 
enforce \emph{use-once} property of the resource allocator.
Moreover, it also implements an \emph{external resource verification} procedure to 
verify if the committed resource is valid. 
}

{\pparagraph{Explicit Time Slot} 
Similar to \postake protocols, 
Filecoin requires processes to explicitly keep track of
the time slot. Hence, in our model, we require the
resource allocator \storageR to maintain the time step number $\slot$ explicitly
and include it to be a part of the resource commitment.}

\begin{definition}[Storage Distribution] The storage distribution $\spaceDis =
\{(p_1, r_1), \dots\}$, where $r_i \geq 0$, specifies the amount of storage that each process pledged through sealing transactions. 
We denote with $\spaceDis^\slot$ the storage distribution at the slot \slot. 
\end{definition}

\begin{definition}[Leader Selection Process]
\label{def:postorageprob}
A \emph{leader selection process} $(\mathcal{D},F)$ with respect to a storage distribution
$\spaceDis= \{(p_1,r_1), \dots\}$ is a pair
consisting of a distribution $\mathcal{D}$ and a deterministic function
${F}$. When $\rho \sample \mathcal{D}$, for all  $sl \in \BN$, $F(\spaceDis, \slot; \rho)$ outputs process $p_i$ with probability
$1 - (1-\dprob)^{r_i} $
where $\dprob$ is the probability that a process gets assigned resource commitment with $r=1$ for a given slot.
\end{definition}

\begin{remark}
we note that in practice Filecoin uses a Poisson sortition\footnote{\url{https://spec.filecoin.io}, Section 4.18.53.111 Running a leader election}, inspired by Algorand cryptographic sortition~\cite{DBLP:conf/sosp/GiladHMVZ17}, to elect multiple leaders. 
This means that a process will pseudorandomly generate a number using its private keys and the value of its random number will determine how many times it is elected leader. The smallest the random value is the more elections it wins. The target values, i.e., the values that determine the number of elections won, are chosen according to a Poisson distribution with parameter $\small{\textsf{numberOfExpectedLeader}}\times \alpha$, i.e., the number of elections won for a process with proportional resource $\alpha$ follows a Poisson distribution of parameter $\small{\textsf{numberOfExpectedLeader}}\times\alpha$. 
If we are to translate this election process to the case where only one leader is elected, we set $\small{\textsf{numberOfExpectedLeader}}=1$ and then a process is elected if it wins at least one election.
As the Binomial distribution can be approximated by Poisson, this is equivalent to using the probability above.
\end{remark}

In addition to the leader selection process, Filecoin protocol also validates the \emph{external} resources pledged by the leader.
This validation is accomplished by requiring the process $p_i$ to provide a cryptographic proof showing the existence of a random data partition (e.g., \postorage).
The randomness used for this challenge is generated through the use of Drand~\cite{DrandDocumentation} values and verifiable random functions.
To simplify the protocol's description, 
similar to the PoS-based protocol, we rely on the resource allocator to securely generate randomness for each slot. Moreover, we require every process $p_i$ to commit $r$ indicating the number of partitions $p_i$ stores.
To capture the \emph{external resource verification} procedure used in Filecoin, we introduce the following definition.

\begin{definition}[External Resource Verification]  
An \emph{external resource verification} procedure $(\mathcal{D},E)$
is a pair consisting of a distribution $\mathcal{D}$ and a deterministic function
${E}$. When $\rho \sample \mathcal{D}$, for every $r$ and  $r'$ in $\BN$ with $r'
\geq r$, $E(r,r'; \rho)$ outputs $\true$ with 
a probability $\displaystyle \frac{r}{r'}$, if $r_i=0$ or $r>r'$, $E()$ outputs \false for all $\rho \in \CD$. 
\end{definition}

\pparagraph{Proof-of-Storage Resource Allocator}
{Proof-of-Storage resource allocator $\storageR$ using the leader selection process $ F()$ 
and the external resource verification $E()$ works as follows. 
First, we require that \storageR keeps track of the current time slot in
order to correctly assign the resource commitment to process $p_i$ for the
current slot. 
Resource allocator \storageR keeps track of the slot through
$\var{Timeout}$ triggered by the $\op{starttimer}()$ event. 
This approach implements the explicit time slot requirement (R1).
Second, upon $\op{\allocator-commit}(p_i,\bstate,r)$ by process $p_i$ with a valid chain
$\chain$ in slot $\slot$, \storageR first checks if 
a randomness $\rho \in \CD$ for \bstate and \slot has been previously sampled; if so, 
then \storageR uses it for leader selection and external resource verification. Otherwise, a fresh randomness is sampled. 
This requirement ensures deterministic and trustworthy source of randomness (R3).
To ensure the leader selection from the common prefix (R2),
\storageR obtains the storage distribution
$\spaceDis$ (i.e., the power table) only from blocks  with slot number less than $sl-k$.
Finally, together with the sampled randomness, \storageR uses $F()$ and $E()$ 
to check if $p_i$ is the leader for
the slot $\slot$ and if $p_i$ passes the external resource verification. 
If both verifications pass then $\storageR$ $\op{assigns}$
the resource commitment to $p_i$, otherwise it $\op{assigns}$ $\bot$. 
If the committed chain $\chain$ is not valid, then $\storageR$
$\op{assigns}$ $\bot$ to $p_i$.}

{A validation of the resource commitment can be done by any process $p_j$
through $\op{\allocator-validate}$; the resource allocator $\storageR$ returns
either \true or \false, depending on the validity of the resource commitment.}

We implement the resource allocator \storageR in Algorithm~\ref{alg:filecoin-pos}.

\begin{algo*}
\vbox{
  \small
  \begin{numbertabbing}
    xxxx\=xxxx\=xxxx\=xxxx\=xxxx\=xxxx\=MMMMMMMMMMMMMMMMMMM\=\kill
    \textbf{state:} \label{}\\
     \> {$B_0$: Genesis block} \label{}\\
     \> $\mathcal{D}$ : Distribution \label{} \\
     \> $F$ : Leader selection function \label{} \\
     \> $E$ : External resource verification function \label{}\\
     \> $\slot$ : Current time slot, initially $\slot=0$  \label{} \\
     \> $k$:~~Common prefix parameter \label{}\\  
     \> $T$ : Set of issued resource commitments, initially empty \label{} \\
     ~\textbf{uses:} \label{}\\
     \> Random oracle: $H:\bset^* \rightarrow \bset^\lambda$ \label{}\\
     ~\textbf{upon} $\op{\allocator-commit}(p_i, \bstate, r)$ \textbf{s.t.} $\bstate = (\chain, B)$, $B=(h, \overline{\tx}, \bot, \bot)$  \textbf{do}: \label{}\\ 
     \> \textbf{if} $\chain$ is valid $\wedge~H(\chain[-1])=h$ \textbf{then} \label{line:post-c-valid} \\ 
      \>\> {obtain $\prechain$ by pruning all blocks with slot $> \slot-k$, from \chain} \label{}\\
      \>\> {\textbf{if} $\prechain = \emptyset$ \textbf{then}}\label{}\\
     \>\>\> {$\prechain \leftarrow [B_0]$}\label{}\\
     \>\> \textbf{if} $\exists (p_i, \prechain, \rho^*, \slot) \in T$ \textbf{then} \` $\rhd$ Queried before \label{}\\
     \>\>\> $\rho \leftarrow \rho^*$  \label{} \\
     \>\> \textbf{else} \label{}\\
     \>\>\> $\rho \sample \mathcal{D}$ \` $\rhd$ Sample a fresh randomness  \label{} \\ 
     \>\>\> $T \gets T \cup \{(p_i, \prechain, \rho, \slot)\}$\` $\rhd$ Update $T$ \label{line:post-T} \\ 
     \>\> {obtain the storage distribution $\spaceDis$ from $\prechain$} \` $\rhd$ Evaluate $\fun{StateAlloc}(\cdot, \cdot)$  \label{} \\
     \>\> {$p_j\becomes F(\spaceDis, \slot; \rho)$}  \label{line:post-Foutputsp} \\  
     \>\> {$b \becomes E(r, \fun{StateAlloc}(p_i, \prechain); \rho)$} \` $\rhd$ Verify the existence of the external resource \label{alg:post-E}\\
     \>\> \textbf{if} $p_i = p_j \wedge b = \true$: \label{line:post-is-leader}\\
    \>\>\> $\pi \leftarrow (p_i, \rho, \slot)$ \label{} \\
    \>\>\> \textbf{output} $\op{\allocator-assign}(p_i, \bstate, \bot, \pi)$ \label{line:post-assign-r}\\
              \>\> \textbf{else}  \label{line:post-no-leader} \\
     \>\>\> \textbf{output} $\op{\allocator-assign}(p_i, \bstate, \bot, \bot)$ \label{alg:post-fail-eval}\\
    \> {\textbf{else} } \label{line:post-c-invalid} \\
     \>\> {\textbf{output} $\op{\allocator-assign}(p_i, \bstate, \bot, \bot)$} \label{line:post-c-invalid}\\
    ~\textbf{upon}  $\op{\allocator-validate}(p_i, \bstate, \pi)$ \textbf{s.t.} $\bstate=(\CC, B), B=(h,\overline{\tx},\pi, \sigma), \pi=(p_j, \rho, \slot)$ \textbf{do}\label{}\\
     \> {obtain $\prechain$ by pruning all blocks with slot $> \slot-k$, from \chain} \label{}\\
     \> {\textbf{if} $\prechain = \emptyset$ \textbf{do}}\label{}\\
     \>\> {$\prechain \leftarrow [B_0]$}\label{}\\
     \> \textbf{if} $\exists (p_i, \prechain, \rho, \slot) \in T$ \textbf{then} \` $\rhd$ Queried before \label{}\\
     \>\> {obtain the storage distribution $\spaceDis$ from $\prechain$} \` $\rhd$ Evaluate $\fun{StateAlloc}(\cdot, \cdot)$  \label{} \\
     \>\> {$p_j^*\becomes F(\spaceDis, \slot; \rho)$}  \label{line:post-Foutputsp} \\  
     \>\> $\var{b} \gets p_j \stackrel{?}{=} p_j^*$ $\wedge~\chain$ is valid $\wedge~H(\chain[-1])=h$ \label{} \\ 
     \>\> \textbf{output} \op{\allocator-is-committed}$(p_i, \bstate, \var{b})$  \label{}\\
     \> \textbf{else}  \label{} \\
     \>\> \textbf{output} \op{\allocator-is-committed}$(p_i, \bstate, \false)$  \label{}\\
     ~\textbf{upon} \op{Timeout} \textbf{do} \` $\rhd$ Increment slot \label{} \\
     \> $\slot \gets \slot + 1$ \label{} \\
     \> \op{starttimer}()  \label{}
    \caption{Implementing \postorage resource allocator, $\storageR$}
    \label{alg:filecoin-pos}
  \end{numbertabbing}
}
\end{algo*}

\begin{lemma}
\label{lem:R_post}
{Given the random oracle $H(\cdot)$, 
the leader selection process $(\CD, F)$ parameterized by the default probability $\dprob$,
the external resource verification $(\CD, E)$,
and the network delay $\Delta$,
there exists a value $R_\CA$ such that
the resource allocator $\storageR$ implemented in Algorithm~\ref{alg:filecoin-pos} is a secure resource allocator.}
\end{lemma}
\begin{proof}
    The \emph{Liveness} property follows from the algorithm: upon $\op{\allocator-commit}(p_i, \var{st}, r)$ from process $p_i$, the resource allocator $\storageR$ either \emph{(i)} has L\ref{line:post-is-leader} satisfied and, eventually, it outputs $\op{\allocator-assign}(p_i, \var{st}, \bot, \pi)$ with a resource commitment $\pi$ to $p_i$, {or it outputs $\op{\allocator-assign}(p_i, \var{st}, \bot, \bot)$ to $p_i$ either {if the external resource verification fails}} or if $p_i$ is not a leader for the current slot (L\ref{alg:post-fail-eval}) or \emph{(ii)} if the chain is invalid (L\ref{line:post-c-invalid}), the allocator outputs $\op{\allocator-assign}(p_i, \var{st}, \bot, \bot)$ to $p_i$.
      
    For the \emph{validity} property, one can apply the same reasoning as in Lemma~\ref{lem:R_pos}; if $p_i$ is a leader for \slot, then {in $T$ there must be the random value $\rho$ previously sampled for $p_i$~(L\ref{line:post-T}). This means that $F$ evaluated on $\rho$ will output again $p_i$. 
    We note that Byzantine process $p$ with no resource can trigger a commit, then uniformly choose $\rho$ later so that $F()$ outputs $p$; however, this happens with a negligible probability for sufficiently large $|\CD|$. 
    }
    Hence, $\pi$ can be validated by any process $p_j$ through $\op{\allocator-validate}(p_j,B, \pi)$; $\storageR$ checks if $\pi \in T$ and it outputs the same result to process $p_j$.
    
  The \emph{use-once} property follows because, in our model, \storageR keeps track of previous $\op{\allocator-commit}$ from $p_i$ along with the time slots and states. 
  Moreover, the choice of $\mathsf{prob}_i$ ensures that an adversary cannot increase its probability of being elected leader by dividing its storage into multiple identities. The proof for this is identical to the proof in~\cref{lem:R_pow}.

  The \emph{unforgeability} property can be proved in the same way as in Lemma~\ref{lem:R_pos}.

  For the \emph{honest-majority assignment}, due to the choice of probability defined in~\cref{def:postorageprob},
  one can derive $\hprob$ and $\aprob$ directly from $R$ and $R_\CA$. 
  In particular, $\hprob = 1-(1-\dprob)^{R-R_{\CA}}$ and $\aprob \leq 1-(1-\dprob)^{R_A}$. 
  Observe that the adversary can commit fewer resources than the pledged amount. However, due to external resource verification $E()$, the probability of getting assigned a resource commitment will be strictly lower. Therefore, $\aprob$ is strictly less than or equal to $1-(1-\dprob)^{R_A}$; 
  the equality occurs only when Byzantine processes commit all resource budgets. 
  Additionally, as in the proof of stake protocols, we note that the adversary can slightly increase $\aprob$ by committing to shorter chains.
  However, in this case, this means that the adversary will fall behind as it has to extend a much shorter chain than the current local chain maintained by correct processes.
  Hence, we consider $\aprob = 1-(1-\dprob)^{R_A}$.
  Thus, we can derive $R_A$ so that $\aprob < \frac{1}{\Delta - 1 + 1/\hprob}$.
\end{proof}

 \begin{theorem}
 \label{thm:postob}
   {Algorithm~\ref{alg:tob} with the secure resource allocator $\storageR$
   implements total-order broadcast.}
 \end{theorem}

 \begin{proof}
   From Theorem~\ref{thm:tob} and Lemma~\ref{lem:R_post}, it follows that, since $\storageR$ is a secure resource allocator, then Algorithm~\ref{alg:tob} with $\storageR$ implements total-order broadcast.
 \end{proof}
 
\section{Trade-offs Between Different Resources}
\label{sec:Sec}

In this section, we describe various attacks against the resource-based total-order broadcast. In particular, we demonstrate long-range attacks against \emph{virtual} resources, and we discuss the incentive consideration that describes the cost of launching attacks against \emph{burnable} and \emph{reusable} resources.  
\subsection{Virtual Resource vs External Resource: Long-Range Attacks}
\label{sub:long-range}
\pparagraph{Long-range Attacks on Virtual Resources} Long-range attacks~\cite{DBLP:journals/access/DeirmentzoglouP19} (LRAs), also sometimes called posterior-corruption attacks, can be mounted on any blockchain based on a virtual resource (such as \postake) if the majority of the set of active processes from an earlier slot becomes inactive in a later slot, as they no longer have any stake left in the system. 
Formally, they can be defined as follows:

\begin{definition}[Virtual-Resource-Shifting Event]~\label{def:shifting}
  {A \emph{Virtual-Resource-Shifting Event} happens when there exist two values $h_0$, $h_1$, and a set of processes $\mathcal{P_{\mathsf{maj}}}$ such that:}
  \begin{itemize}
    \item {\textbf{Active at $h_0$}: At height $h_0$, processes in
        $\mathcal{P}_{\mathsf{maj}}$ control the majority of the total
      virtual resource (i.e., $R$), namely:}
       ${\sum_{p_i \in \mathcal{P}_{\mathsf{maj}}} (\fun{StateAlloc}(p_i,\chain[0:h_0])) > R-R_A}$
    \item {\textbf{Inactive at $h_1$}: At height $h_1>h_0$, processes in $\mathcal{P}_{\mathsf{maj}}$ control
      less virtual resources than the total number of resources controlled by the adversary (i.e., $R_{\mathcal{A}}$):} 
        ${\sum_{p_i \in \mathcal{P}_{\mathsf{maj}}} (\fun{StateAlloc}(p_i,\chain[0:h_1])) \leq R_{\mathcal{A}}}$
  \end{itemize}
\end{definition}
If \cref{def:shifting} is satisfied, then most processes
in $\mathcal{P}_{\mathsf{maj}}$ have released all or part of their resources
by height $h_1$, and the adversary has enough budget to 
corrupt all the active processes in $\mathcal{P}_{\mathsf{maj}}$ since they
are all inactive in the present. The adversary could then use these processes
to re-write the chain from $\chain[h_0]$ since with a virtual
resource as no external resource is needed to call the resource allocator. 
Furthermore, the \emph{release of resource} from processes in
$\mathcal{P}_{\mathsf{maj}}$  also happens on-chain, e.g., in the case of
\postake for a process to move from active to inactive, it will spend its
coins on-chain. An adversary re-writing the history of the chain could simply
omit these transactions such that all processes satisfying
definition~\ref{def:shifting}
stay active in the alternative chain that the adversary is writing.
The attack proceeds as follows:
\begin{enumerate}
    \item When a virtual-resource-shifting event happens at the current height $h_1$, $\mathcal{A}$
      corrupts all processes in $\mathcal{P}_{\mathsf{maj}}$. 
      Since the total of resources controlled by these processes is less than $R_\CA$, $\CA$ has enough budget
      to do so;
    \item\label{enum:lra-grow-chain} $\mathcal{A}$ starts a new chain
      $\chain^*$ at $\chain[h_0]$. At this height, $\CA$ controls the majority
      of the virtual resource, and because the resource allocator takes
      no further input apart from the state of $\chain[0:h_0]$, it assigns the
      resource commitment to $\mathcal{A}$ with high probability;
    \item  $\mathcal{A}$ now controls all processes in
      $\mathcal{P}_{\mathsf{maj}}$ and can alter the state of the chain such
      that the processes in $\mathcal{P}_{\mathsf{maj}}$ never release
      their resource; 
    \item The adversarial chain will grow at a faster rate and will \emph{eventually} become
      longer than the honest chain because there is no network delay between
      corrupted processes. 
\end{enumerate}
\pparagraph{Long-range Attacks on External Resources} The strategy above does not work with \emph{external} resources.
Even if \cref{def:shifting} holds, the adversary cannot call the
resource allocator by simply corrupting the processes $p_i$ as an \emph{external}
resources would be needed as input to the resource allocator
(step~\ref{enum:lra-grow-chain} in the strategy above).

We formalize the implication of the long-range attack in the following lemma and theorem.

\begin{lemma}[Long-range Attack]\label{lem:lra}
  In a virtual-resource-based total-order broadcast (Algorithm~\ref{alg:tob}),
  let $\lchain$ be the longest chain maintained by a correct process, if a virtual-resource-shifting event occurs, then an adversary can \emph{eventually}
  form a valid chain $\chain^*$ that is {longer than} $\lchain$.
\end{lemma}
\begin{proof}
  If a virtual-resource-shifting event (\cref{def:shifting}) occurs, 
  an adversary $\mathcal{A}$ can corrupt all
  $p_i \in \mathcal{P}_{\mathsf{maj}}$ at height $h_1$. 
  Notice that $\mathcal{A}$ can do this because according to the threat model
  defined in definition~\ref{def:corruption}, $\CA$ has enough resource budget
  to corrupt all $p_i \in \mathcal{P}_{\mathsf{maj}}$.

  Adversary $\mathcal{A}$ can start a new chain $\chain^*$ at height $h_0$
  by requiring all the corrupted processes $p_i \in
  \mathcal{P}_{\mathsf{maj}}$ to commit old states to $\stakeR$. 
  Since the Byzantine processes control the majority of the resources, 
  the probability of Byzantine processes getting assigned commitment is strictly higher than 
  the probability of correct processes getting assigned commitment; hence,
  the growth rate of $\chain^*$ is strictly higher than the growth
  rate of the honest chain $\lchain$; therefore,  
  $\chain^*$ will eventually catch up and outgrow $\lchain$ in terms of
  the length. 
  
  {More concretely, to simplify our analysis, we also assume the network delay to be $1$ (i.e., $\Delta = 1$) between correct processes.
  We recall that $\hprob$ is the probability that at least one correct process gets selected on the honest chain $\lchain$ at each time step. 
  For any interval $[t_0,t_0+t]$ and arbitrary $t_0, t\in \BN$, 
  we denote with $X_0, \dots, X_{t-1}$ independent Poisson trails such that $\Pr[X_i=1] = \hprob$, 
  and we let $X_H = \sum_{i=0}^{t-1} X_i$.
  Using the Chernoff bound, one can show that for any $\epsilon \in (0,1)$ it holds that $\Pr[X_H < (1-\epsilon)\cdot \hprob\cdot t] \leq \exp(-\hprob \cdot t\cdot \epsilon^2/2)$. 
  Intuitively, the Chernoff bound implies that the value of $X_H$ cannot deviate too much from the mean; hence, for sufficiently large $t$ and  sufficiently small $\epsilon$,
  the upper bound on the honest chain growth is approximately $\hprob\cdot t$, with an overwhelming probability.}
  
  {Using the same argument for the growth of the malicious chain, one can show that
  for a sufficiently large time interval (i.e., $t$) and a sufficiently small $\epsilon$, 
  the lower bound of chain growth is approximately $\aprob \cdot t$ (i.e., $(1+\epsilon)\cdot \aprob \cdot t$) with an overwhelming probability (i.e., $\exp(-\aprob\cdot t \cdot \epsilon^2/3)$), where $\aprob$ is the probability that at least one Byzantine processes get selected on the honest chain $\chain^{*}$ at each time step.}

  {So, if $\aprob > \hprob$, 
  we can claim that $\chain^{*}$ grows at a faster rate than $\lchain$. This is the case for Algorithm~\ref{alg:tob} that uses \stakeR allocator. Due to~\cref{def:shifting}, 
  the probability of getting assigned the resource commitment with $\chain^*$, is $\aprob > 1-(1-\dprob)^{R-R_\CA} = \hprob$, where $\hprob$ is the probability that at least one correct processes get assigned a resource commitment with \lchain.}
  \end{proof}

\begin{remark}
  {Also, if we assume a $\Delta > 1$ network delay between correct processes, there will a non-zero probability that a fork can happen, and honest blocks can get discarded due to
  the network delay. On the other hand, we also assume a perfect synchrony
  ($\Delta=1$) between Byzantine processes; therefore, there is no loss in the
  malicious growth rate. Therefore, 
  even when correct processes and Byzantine processes control the same amount of resources on both chains, due to network delay, the chain growth rate of $\chain^*$ can still be 
  higher than the chain growth rate of the honest chain $\lchain$}.
\end{remark}

\begin{theorem}\label{theo:pos-safety}
  If a virtual-resource-shifting event occurs, a total-order broadcast based on
  virtual resources (Algorithm~\ref{alg:tob}) does not guarantee properties of total-order
  broadcast.
\end{theorem}

\begin{proof}
Let $\chain$ be the honest chain adopted by every correct process and let us
assume that all the transactions buried at least $k$ blocks deep in $\chain$
have been $\op{a-delivered}$ (Algorithm~\ref{alg:tob},
L\ref{algo:pox-a-deliver}) by every correct process.
  If a virtual-resource-shifting event occurs then, by Lemma~\ref{lem:lra}, an
  adversary can eventually
  form a valid chain $\chain^*$ that is {longer than} $\chain$.

  For existing processes, the adversary can send this $\chain^*$ to a subset of correct processes. 
  This implies that some correct processes will adopt $\chain^*$ as a valid chain; they will $\op{a-deliver}$ all the transactions buried at least $k$ blocks deep in $\chain^*$.
  {This implies that, eventually, the \emph{total-order} property is violated.}
  Also, due to the permissionless nature of our model, correct processes might join the system at any time. Hence, new processes will adopt the malicious chain as the local chain; therefore, \var{delivered} will be different among correct processes. 
  Hence,  the \emph{total-order} property is violated
\end{proof}

\subsection{Incentives in Burnable and Reusable Resources}
\label{sub:incentive}
One of the vulnerabilities induced by \emph{reusable} resources
is that extending the blockchain is \emph{costless} with
respect to the resource considered.
This is different from \emph{burnable} resources, where creating a block consumes
the resource; {this consumption is captured in our model
as the interaction between processes and the resource allocator}.
The use of reusable resources can result in two different types of adversarial behaviors. 
The first one consists in creating multiple blocks at the same time slots on
different chains. 
The second one consists in keeping blocks created private from the rest of the
processes.
In both cases, we discuss how this \emph{costless} property associated with block
creation for longest-chain consensus protocols based on a reusable resource
impacts their security compared to those based on burnable resource.
In this section, we assume, as is traditional with any blockchain system, that
some financial reward is associated with block creation,
and we assume the cost of acquiring
resources is the same for both reusable and burnable resources. With these
assumptions and the use-once property of resource allocator, we define the chain
extension cost as follows.

\begin{definition}[Chain Extension Cost]
\label{def:chain-extension-cost}
{The cost of extending a valid chain for a process $p_i$ between two time steps 
$t_1$ and $t_2$ such that $t_1 \leq t_2$ is defined to be the resource budgets
committed and assigned back during this time interval. In particular, we have:}
\begin{itemize}
  \item 
  {For burnable resources:  
  $\cost_{\mathsf{burn}}(p_i, t_1, t_2) = \sum_{t=t_1}^{t_2}\fun{Alloc}(p_i,t)$}
  \item  {For reusable resources:
  $\cost_{\mathsf{reuse}}(p_i,t_1,t_2) \leq \max_{t \in [t_1, \dots, t_2]}\{\fun{Alloc}(p_i,t)\}$}
\end{itemize}
\end{definition}
\begin{proposition}~\label{prop:cost}
  {For all time step $t_2>t_1$ and a process $p_i$, the cost of extending a valid chain with a burnable resource is strictly more expensive than with a reusable resource, i.e., $\cost_{\mathsf{burn}}(p_i,t_1, t_2)>\cost_{\mathsf{reuse}}(p_i, t_1,t_2)$.}
\end{proposition}

{\cref{prop:cost} indicates that it is inherently more expensive to extend the blockchain for \emph{burnable} resources; hence, it is more difficult to launch different types of attacks on blockchain based on \emph{burnable} resources. In the following, we explain different types of attacks.}

\pparagraph{Private Attack} The private attack~\cite{DBLP:conf/ccs/DemboKTTVWZ20}, sometimes called double-spending
attack, is the most simple attack in longest-chain blockchains.
The adversary creates a private chain, i.e., it mines on its own without
broadcasting its blocks to the other processes and without accepting the blocks
from other processes.
In particular, the adversary runs Algorithm~\ref{alg:tob}, except that it does
not broadcast its blocks until the end of the attack.
This means that two chains grow in parallel: the adversarial one, that only the
adversary is aware of, and the honest one.
The adversary is aware of the honest chain but chooses not to contribute to it
and it wins the attack if it creates a chain longer than the honest chain.
In the case of a \emph{burnable} resource, this attack has a cost as every block
created consumes a resource. If the adversary wins the attack, then the cost is
recovered as the adversary wins the reward associated with block creation.
Otherwise, it loses the cost associated with all the resources consumed.
In the case of a \emph{reusable} resource, the only cost of the attack is the
\emph{opportunity cost}, i.e., the adversary takes the risk of potentially not
earning the rewards associated with block creation if the attack fails but does
not lose any resources. The attack in this case is then much cheaper than in the 
case of a burnable resource.
The cost of a private attack is higher if the resource allocator is based on a
burnable resource than if it is on a reusable resource, thus creating a
stronger disincentivisation for an adversary.
The results follow from the fact that for a reusable resource, the resource
allocator can be invoked 
on the same resource several times. 
From \cref{prop:cost}, it is not difficult to see that the expected return on performing a private attack is higher for a \emph{reusable} resources as the probability of winning the attack (i.e. producing a longer chain) is the same in both cases, but the cost is higher for a \emph{burnable} resource. 

\pparagraph{Resource-bleeding Attack}
Stake-bleeding attacks~\cite{DBLP:conf/cvcbt/GaziKR18} were proposed in the context of \postake blockchains and work, informally, as follows. An adversary starts creating a private chain (i.e., it does not broadcast its blocks to the rest of the network) but, differently from the private attack described previously, the adversary may continue creating blocks on the honest chain. In its private chain, the adversary includes all of the transactions it is aware of, harvesting the associated transaction fees. Furthermore, the adversary also receives the coinbase reward usually associated with block production.
After a sufficient amount of time, the adversary will have bloated its amount of resources and will \emph{eventually} be able to create a chain that becomes as long as the honest chain.
This attack could be extended to the general-resource case, which we call this attack \emph{resource-bleeding attack}, and note that in the case of an external resource, this attack is much easier to detect than in the Proof-of-Stake case.
In order to understand this attack, we must extend the model from Section~\ref{sec:nakamoto_consensus}
to take into account \emph{total resource adjustments} in the case of inactive processes.
In Appendix~\ref{sec:discussion} we describe the general case of resource-bleeding attacks and discuss how they are more detectable on an external resource and the most mitigated for burnable resources.

\pparagraph{Nothing-at-Stake Attack} In a Nothing-at-Stake attack, instead of
deciding to extend the longest chain (Algorithm~\ref{alg:tob},
L\ref{algo:pox-longest-chain}), a process decides to mine simultaneously on all
of the chains it is aware of.
In the case of a burnable resource, an adversary cannot reuse the same resource
to mine on multiple chains (due to L\ref{line:resource-burnt}), hence in order
to mount this attack, the adversary must decide how to commit its resources to
multiple chains. 
In contrast, with a reusable resource, each resource can be fully committed to each chain. If there exists multiple forks of the same length,
there is a risk that a process will mine on a chain that ends up being abandoned and thus will miss out on the associated reward.
It thus becomes rational for a process to deviate from the protocol and mine on every chain since this reduces chance of losing reward because network may select different chain.
If every process adopts this strategy, the protocol cannot achieve the common prefix property as every chain will keep on growing at the same pace.

\section{Discussion}
\label{sec:discussion}
\pparagraph{Resource-bleeding Attack in the Flexible Resource Setting}
The resource-bleeding attack stems from this observation: in order to deal with inactive processes, if the protocol wants to maintain its block production rate, it needs to adjust its leader selection processes such
that inactive processes are not selected anymore. In practice, this means increasing $\prob$ such that every active process has a higher chance of being selected and removing the inactive processes from the list of eligible block producers and hence maintaining a steady block rate.
If an adversary starts a private attack, since no resource commitment from the other processes is included in the adversarial chain, after a sufficient time, $\prob$ will be updated to ensure that the adversarial chain block rate is maintained.
On the other hand, with a reusable resource, the adversary could keep maintaining its resources on the honest chain to ensure that the leader selection probability is not adjusted on the honest chain. 
After enough time, all the honest processes will be removed from the power table in the adversarial chain. This means that when electing a leader (Algorithm~\ref{alg:filecoin-pos}, L\ref{line:post-Foutputsp}) on the adversarial chain, the adversary now represents the full power table and is guaranteed to be elected at each epoch.
On the other hand, since the adversary maintain its resource on the honest chain, without contributing as many blocks as it could. This means that, after some time, the honest chain will grow at a slower rate than the adversarial chain and the adversary will be able to create a chain as long as the honest chain, breaking the safety of the protocol.

In practice in Bitcoin, the \emph{target} value~\cite{bitcointarget} is updated every two weeks (roughly) to ensure that blocks are created, on average, at the same pace. An adversary could fork the chain, wait for the difficulty adjustment to adjust and then be able to create a chain at the same pace as the honest chain. This is, however, easily detectable. In the PoW case, one can simply see that the difficulty has been adjusted and that one chain has much fewer resources than the other.
Moreover, since the resource is burnable, it is not possible for an adversary to continue mining on the honest chain as the same burnable resource cannot be used twice, hence the adversary cannot maintain its full resource on the honest chain and the honest chain difficulty must be adapted accordingly.

For an external, but reusable, resource such as storage, the adversary could maintain its power in both chain, however, it is easy to detect the adversarial chain as
it will have fewer resources committed to it and hence is distinguishable from the honest chain.

\pparagraph{Mitigations against Different Attacks}
In the following, we discuss various mitigations against attacks described 
in~\cref{sec:Sec}.

\textit{Long-range Attacks.} {In practice, many \postake systems deal with long-range attacks by using some form of checkpointing~\cite{DBLP:conf/aft/AzouviDN20,DBLP:journals/iacr/TasTGKMY22,DBLP:journals/corr/abs-2109-03913}, 
requiring key-evolving cryptography~\cite{DBLP:conf/eurocrypt/DavidGKR18,DBLP:conf/sosp/GiladHMVZ17}, 
or using multiple types of resources~\cite{DBLP:conf/ccs/FitziWKKLVW22}. 
Others use more refined chain selection rules~\cite{DBLP:conf/eurocrypt/DavidGKR18,DBLP:conf/ccs/BadertscherGKRZ18} (i.e., chain density analysis or selecting the longest chain that fork less than $k$ blocks) instead of the longest chain selection.}

\textit{Resource-bleeding Attacks.}
In the case of \postake, mitigation has been proposed in Ouroboros Genesis~\cite{DBLP:conf/ccs/BadertscherGKRZ18} and it works as follows. When a process is presented with two forks, it differentiates between two cases. In the first case, the fork is smaller than the common prefix parameter $k$, i.e., the two chains differ for a number of slots smaller than $k$, in which case the usual longest-chain rule is applied. If on the other hand, the forks differ from more than $k$ slots, then the processes look at the first $k$ slots after the fork (i.e., the first $k$ slots where the two chains diverge) and choose the chain with the most blocks in that period. Intuitively, this is because during the beginning of the fork, an adversary has not had the time to bloat its stake and hence the rate at which its chain grows will be smaller than that of the correct processes.
In the case of an external resource, it suffices to look at the total power (which can be explicit in the case of a reusable resource, or implicit for a burnable resource, e.g., \emph{target} value) at the tip (end) of the chains and pick the one with the most resource.

\textit{Nothing-at-stake Attacks.}
A process that performs a nothing-at-stake attack with a reusable resource is
easily detectable as anyone can see that the same resource was used on
different chains. 
One typical mitigation adopted by
\postake systems is to \emph{slash}, i.e., financially punish, processes who
use their resource on concurrent chains. This is usually done by having
processes deposit some money before gaining participation rights, and then
burning some of this deposit if a proof of misbehavior is sent to the
blockchain. The details of this mechanism are out of scope for this paper.

\section{Conclusion}\label{sec:conclusion}
Resources are essential in ensuring the safety property of total-order
broadcast protocols in a permissionless setting as it protects the protocol from Sybil attacks. 
However, there exist several attacks on protocols based on reusable and virtual resources that 
a formal specification would help understand and address.
  
{In this work, we formalize properties of resources through a
\emph{resource allocator} abstraction, and identify crucial properties on
how to make this resource allocator secure for blockchain protocols. 
Using a secure resource allocator, we demonstrate how to construct a generic
\emph{longest-chain} \emph{total-order} broadcast algorithm. 
Furthermore, we also illustrate how certain types of resources tend to make blockchain protocols 
more vulnerable to different types of attacks. 
We believe that this formalization
will help blockchain protocol designers to select suitable types of resources for
their protocols and understand and analyze the potential security trade-offs on those resources.  
}

\pparagraph{Outlook} 
{For future work, we find the following research
directions worth investigating:} 
\begin{itemize}
    \item {\textbf{Relaxed Assumptions.} Our analysis works with a setting where the total amount of active resources is known and fixed. Hence, it is natural to extend this model to a setting where the total amount of resources is unknown and potentially fluctuates.} 
    \item \textbf{Different Network Setting and Participation Models.} Our model focuses on \emph{probabilistic} \emph{longest-chain} protocols in a $\Delta$-synchrony setting. However, we believe that our model can be applied to analyze properties of resource-based \emph{deterministic} protocols in a permissioned and partially synchrony setting such as Tendermint~\cite{DBLP:journals/corr/abs-1807-04938} and HotStuff~\cite{DBLP:conf/podc/YinMRGA19}.
    \item {\textbf{Different Types of Resources.} Finally, there are other resource-based protocols such as the Proof-of-Elapsed-Time~(PoET) protocol~\cite{DBLP:conf/indocrypt/BowmanDMM21} or multi-resources-based protocol~\cite{DBLP:conf/ccs/FitziWKKLVW22} that have not been considered in this work. 
    Hence, one can extend this model to analyze those protocols.}
\end{itemize}

\section*{Acknowledgments}

The authors thank anonymous reviewers for helpful feedback.
DVL has been supported by a grant from Protocol Labs to the University of
Bern.  LZ has been supported by the Swiss National Science Foundation (SNSF)
under grant agreement Nr\@.~200021\_188443 (Advanced Consensus Protocols).

\bibliography{dblpbibtex, references}
\end{document}